\documentclass[11pt,reqno]{amsproc}

\title[Modelling climate and weather of 2D LA-SALT Euler-Boussinesq]{Modelling the climate and weather of a 2D Lagrangian-averaged Euler-Boussinesq equation with transport noise}

\author[D. Alonso-Or\'an]{Diego Alonso-Or\'an}
\author[A. Bethencourt de Le\'on]{Aythami Bethencourt de Le\'on}
\author[D. Holm]{Darryl Holm}
\author[S. Takao]{So Takao}

\address{Instituto de Ciencias Matem\'aticas CSIC-UAM-UC3M-UCM, 28049 Madrid, Spain.}
\email{diego.alonso@icmat.es}
\address{Department of Mathematics, Imperial College, London SW7 2AZ, UK. }
\email{ab1113@ic.ac.uk}
\address{Department of Mathematics, Imperial College, London SW7 2AZ, UK. }
\email{d.holm@ic.ac.uk}
\address{Department of Mathematics, Imperial College, London SW7 2AZ, UK. }
\email{st4312@ic.ac.uk}

\usepackage{amsmath,latexsym,amssymb,verbatim,amsbsy,amsthm,mathrsfs,tikz,times}
\usepackage[margin=1in]{geometry}
\usepackage[shortlabels]{enumitem}
\usepackage{empheq}
\usepackage[color=yellow]{todonotes}
\usepackage{enumerate}
\usepackage{appendix}
\usepackage{color}
\usepackage{parskip}
\usepackage{bigints}
\usepackage[colorlinks=true, pdfstartview=FitV, linkcolor=blue, citecolor=black, urlcolor=blue]{hyperref}

\theoremstyle{plain}
\newtheorem{theorem}{Theorem}[section]
\newtheorem{definition}[theorem]{Definition}
\newtheorem{lemma}[theorem]{Lemma}
\newtheorem{proposition}[theorem]{Proposition}

\theoremstyle{definition}
\newtheorem{remark}[theorem]{Remark}
\newtheorem{example}[theorem]{Example}

\numberwithin{equation}{section}

\newcommand{\overbar}[1]{\mkern 1.5mu\overline{\mkern-1.5mu#1\mkern-1.5mu}\mkern 1.5mu}

\def\tilde{\widetilde}

\renewcommand\hat{\widehat}

\def\RR{{\mathbb R}}
\def\TT{{\mathbb T}}

\def\LL{{\mathcal L}}

\def \u {\mb{u}}
\def \f {\mb{f}}
\def \U {\mb{U}}
\def \V {\mb{V}}

\def \xxi {\mb{\xi}}
\def \x {\mb{x}}
\def \y {\mb{y}}
\def \X {\mb{X}}
\def \A {\mb{A}}
\def \div {\mathrm{div}}
\def \p{\partial}

\DeclareMathOperator{\diff}{d\!}
\newcommand{\norm}[1]{\left\lVert#1\right\rVert}    
\newcommand\abs[1]{\left|#1\right|}    

\newcommand{\rmd}{{\rm d}}
\newcommand{\scp}[2]{{\big\langle {#1}\, , \, {#2}\big\rangle}}
\newcommand{\Scp}[2]{{\Big\langle {#1}\, , \, {#2}\Big\rangle}}
\newcommand{\SCP}[2]{{\left\langle {#1}\, , \, {#2}\right\rangle}}

\newcommand{\mb}[1]{\mbox{\boldmath{$#1$}}}
\newcommand{\mbs}[1]{\footnotesize{\mbox{\boldmath{$#1$}}}}

\newcommand{\E}[1]{\mathbb{E}\left[{#1}\right]}
\begin{document}

%%%%%%%%%%%%%%%%%%%%%%%%%ABSTRACT%%%%%%%%%%%%%%%%%%%%%%%%%%%%%%%%%%%%%%%%%%%%%%%%%%
\begin{abstract} 
The prediction of climate change and its impact on extreme weather events is one of the great societal and intellectual challenges of our time. The first part of the problem is to make the distinction between weather and climate. The second part is to understand the dynamics of the fluctuations of the physical variables. The third part is to predict how the variances of the fluctuations are affected by statistical correlations in their fluctuating dynamics. This paper investigates a framework called LA SALT which can meet all three parts of the challenge for the problem of climate change.  As a tractable example of this framework, we consider the Euler--Boussinesq (EB) equations for an incompressible stratified fluid flowing under gravity in a vertical plane with no other external forcing. All three parts of the problem are solved for this case. In fact, for this problem, the framework also delivers global well-posedness of the dynamics of the physical variables and closed dynamical equations for the moments of their fluctuations. Thus, in a well-posed mathematical setting, the framework developed in this paper shows that the mean field dynamics combines with an intricate array of correlations in the fluctuation dynamics to drive the evolution of the mean statistics. The results of the framework for 2D EB model analysis define its climate, as well as climate change, weather dynamics, and change of weather statistics, all in the context of a model system of SPDEs with unique global strong solutions. 

\hfill \today
\end{abstract}
\maketitle
\renewcommand\contentsname{}
\tableofcontents
%%%%%%%%%%%%%%%%%%%%%%%%%%%%%%MAIN%%%%%%%%%%%%%%%%%%%%%%%%%%%%%%%%%%%%%%%%%%

\section{Introduction}
\paragraph{\bf Background.} To meet the challenge of climate change prediction in practice, one must predict the coarse-grained dynamic changes of an extremely complex atmosphere/ocean system which is only partially observed by using a suite of imperfect theoretical and computational simulation models. This means that predictions of quantities of climate interest may be strongly affected by uncertainty arising from unknown model errors and incomplete knowledge of state variables. In addition, one must assess the impacts of climate change over a wide range of significant temporal and spatial scales. For example, one must predict and understand the seasonal, yearly, decadal, and centennial impacts of climate change for issues ranging from extreme weather events, to sea level rise, and the dynamic distributions of deserts and forests.

\paragraph{\bf Previous approaches.} Deterministic physics characterises the climate change problem as a high-dimensional complex dynamical system with sensitivity to initial conditions on essentially all spatial and temporal scales. To estimate the level of difficulty of the climate change problem, one notes that the turbulence problem falls into this same class of problems. The governing Navier-Stokes equations are known for turbulence, though. The central difficulty of climate change science is that the dynamical equations for the actual climate are unknown. In fact, even the definition of climate is still under discussion in the literature \cite{Bothe2018}.  

As in turbulence theory, the statistical approach to the climate system has been developed in parallel to the deterministic computational approach. This development goes back at least fifty years to the early predictability studies for simplified atmosphere models \cite{Epstein1969, Lorenz1963, Lorenz1965, Lorenz1969, Lorenz1976, Lorenz1995, Lorenz1996}.

In a celebrated unpublished paper \cite{Lorenz1995} Ed Lorenz defined the statistical approach to climate science by quoting the following adage.
\begin{quote}
``Climate is what you expect. Weather is what you get." 
\end{quote}
This adage captures the essence of the problem. Namely, climate science is fundamentally probabilistic. 

In the same unpublished paper \cite{Lorenz1995} Lorenz remarked that:
\begin{quote}
There are many questions regarding  climate whose answers  remain  elusive. For example, there is the question of determinism;  was it somehow inevitable at some earlier time that  the climate now would  be as it actually is?
\end{quote}

To address some of his questions in \cite{Lorenz1995} and particularly to address climate change without giving up determinism, Lorenz postulated the idea of an ``almost intransitive'' dynamical system, as follows.
\begin{quote}
An {\it almost intransitive} system is one that can undergo two or more distinct
types of behaviour, and will exhibit one type for a long time, but not forever.
\end{quote}
Since then, many people have discussed this issue, especially as it has become increasingly urgent. A recent  review appears, e.g., in \cite{DaSt2013}. Lorenz seemed to suggest in  \cite{Lorenz1995} that the expected solution itself could be almost intransitive. Answering this  question would require a deterministic equation for the expected solution.

Lorenz's concept of ``almost intransitivity'' also recalls the concept of \emph{intermittency} discussed in turbulence modelling using the Navier--Stokes (NS) equations, although intermittency is usually regarded at the shorter time scales available for typical turbulence problems. 

%While turbulence is known to be governed by the underlying NS equations, their well-posedness remains an open question. 

Computational simulation of Navier--Stokes turbulence faces a \emph{closure problem}, because it is unable to encompass all of the spatial and temporal scales which develop in the turbulent cascade of energy.  Climate science faces an even more extensive closure problem, if it makes the assumption that the weather and the climate obey the same equations. 
The question then arises, ``Would turbulence modelling approaches apply to the climate, if the climate were defined as simply `what you expect' as a statistical property of a dynamical system?" 

This turbulence question engages another recently developing computational approach in climate/weather numerical simulations. This approach involves the introduction of \textit{stochastic parameterisation}, in which mean quantities of interest do have a precise sense of `expectation' and the remainder at a given instant has a sense of `fluctuation'. For recent reviews of this approach, see, e.g. \cite{BYP2012, Berner-etal-2017,GCF2015}.   In the approach to stochastic parameterisation, the summary conclusion of \cite{BYP2012} is that
\begin{quote}
\textit{a posteriori} addition of stochasticity to an already tuned model is simply not viable.

This in turn suggests that stochasticity must be incorporated at a very basic level within the design of physical process parameterisations and improvements to the dynamical core.
\end{quote}
One approach in line with this conclusion is the SALT (stochastic advection by Lie transport) approach introduced in \cite{Holm2015}. The SALT approach combines stochasticity at the `basic level' of Kelvin's circulation theorem, along with the particle filtering method used for data assimilation. A protocol for applying the SALT approach in data assimilation based on comparing fine scale and coarse scale computational simulations has recently been developed in \cite{CCHOS18a,CCHOS18b}. The rest of the present paper will concentrate on developing a Lagrangian-averaged (LA) version of SALT which was recently proposed in \cite{DrivasHolm2019} and developed further in \cite{DHL2019} for potential use in climate change science.

\subsection*{Aims of the present paper.}
In this paper, we derive a stochastic version of the two-dimensional Euler-Boussinesq fluid system which is non-local in \emph{probability space}, rather than in physical space, in the sense that the expected velocity is assumed to replace the drift velocity in the transport operator for the stochastic fluid flow. This stochastic fluid model is derived by exploiting a novel idea introduced in \cite{DrivasHolm2019}, of applying Lagrangian-averaging (LA) in \emph{probability space} to the fluid equations governed by stochastic advection by Lie transport (SALT) which were introduced in \cite{Holm2015}.  

We follow the LA SALT approach to achieve three results of interest in climate modelling based on the Kelvin circulation theorem for stochastic transport of the Kelvin loop. The three results address the three components of the climate change problem discussed at the outset. First, it answers Lorenz's question about determinism in the affirmative. Namely, by replacing the drift velocity of the stochastic vector field by its expected value, one finds that the expected fluid motion becomes deterministic. This first step leads to the second result of interest in climate change modelling. Namely, it reduces the dynamical equations for the fluctuations to a linear stochastic transport problem with a deterministic drift velocity.  Such problems are well-posed. We prove here that the LA SALT version of the 2D EB problem in a vertical plane possesses global strong solutions. The third result addresses the dynamics of the variances of the fluctuations. This result demonstrates that the variances and higher moments of the fluctuation statistic evolve deterministically, driven by a certain set  of correlations of the fluctuations among themselves. Thus, the first result of the paper makes the distinction between climate and weather for the case at hand. Namely, the LA SALT fluid equations for 2D EB may be regarded as a dissipative system akin to the Navier--Stokes equations for the expected motion (climate) which is embedded into a larger conservative system which includes the statistics of the fluctuation dynamics (weather). 
The second result provides a set of linear stochastic transport equations for predicting the fluctuations (weather) of the physical variables, as they are driven by the deterministic expected motion. The third result produces closed deterministic evolutionary equations for the evolution of the variances and covariances of the stochastic fluctuations and their $p$-th order central moments in certain cases. 

In summary, the 2D EB model system treated here by the LA SALT approach reveals that its statistical properties are fundamentally dynamical. The results of the 2D EB LA SALT model analysis define climate, as well as climate change, weather, and change of weather statistics, all in the context of a model system of SPDEs with unique global strong solutions.

\subsection*{Plan of the paper}$\,$

Section \ref{sec:EBsystem} introduces the 2D EB LA SALT system and computes the dynamics of the expectation and fluctuation components of its solutions, as well as their variances. 

Section \ref{sec:LASDPsystems} computes expectation and fluctuation dynamics for LA SALT equations, as well as their variances, covariances and $p$-th central moments, in a general setting. In general, the dynamics of these statistics for LA SALT does not close. However, the fluctuation statistics for the 2D EB LA SALT system in fact does close and the properties resulting from this closure are discussed in Example \ref{2DEB-example}. 

Section \ref{sec:prelim-anal} describes the analytical setting and explains the approach in obtaining the main result Theorem \ref{main:th:1} of well-posedness of the 2D EB LA SALT system, as proved in Section \ref{sec:well-posedness}, subsection \ref{Thm-1-proof}. Namely, for sufficiently smooth initial conditions $(\u_0,\theta_0) \in H^2(\mathbb T^2, \mathbb R^2) \times H^3(\mathbb T^2, \mathbb R)$,  there exists a unique global strong solution of the 2D LA-SALT EB equations \eqref{LA:SALT:Ito:Bou}.

\subsection*{Acknowledgements}
The authors thank Theodore Drivas, James-Michael Leahy, Dan Crisan and Wei Pan for stimulating and encouraging discussions. DDH is grateful for partial support by the EPSRC Standard Grant EP/N023781/1.
ST acknowledges the Schr\"odinger scholarship scheme for funding during this work. DAO acknowledges financial support from the Spanish Ministry of Economy and Competitiveness, through the “Severo Ochoa Programme for Centres of Excellence in R\&D” (SEV-2015-0554)”. ABdL acknowledges PhD student support from the EPSRC Centre for Doctoral Training, ``Mathematics of Planet Earth". 

\section{The Euler-Boussinesq (EB) fluid system in a vertical plane}\label{sec:EBsystem}

In concert with the idea that the climate should be computed with the same fundamental equations as the weather, this paper addresses a representative model of stratified incompressible flow which is a component of any climate model. Namely, it addresses the familiar Euler-Boussinesq (EB) fluid system in a vertical plane. The issue of global existence of regular solutions of the deterministic Boussinesq model still remains an outstanding open problem. Its SALT version inherits most of the properties of its deterministic counterpart and its local well-posedness has been recently established in \cite{DieAytJNLS}.  We first recall the introduction into 2D EB of stochastic advection by Lie transport (SALT) as discussed in that work. We then apply the Lagrangian averaging (LA) concept in probability space to derive and analyse the LA SALT version of the 2D EB equations.  We establish global well-posedness of the LA SALT EB system and investigate the solution behaviour of this stochastic PDE system. 

We begin with the following question. What is the Kelvin circulation theorem for the 2D EB climate/weather system?

\subsection{What is the Kelvin circulation theorem for the 2D EB climate/weather system?}
%If climate is what you expect, it should be represented by the expectation of something about the climate/weather fluid system (\textcolor{red}{A: this last phrase seems a bit too vague ("the expectation of something about the climate")}. 
%The key feature of interest about the climate/weather system is its circulation. So, the question arises, ``What is the Kelvin circulation theorem for the climate?'' 

The Kelvin circulation theorem is a statement of Newton's Force Law for the motion of distributions of mass on closed material loops $c(\mb{u}_t^L)$, where the subscript $t$ denotes explicit time dependence. By definition, such material loops move with the transport velocity $\mb{u}_t^L$ of the fluid flow. Newton's Force Law states that the time rate of change of the momentum $\mb{P}$ of such a loop of a given mass distribution is equal to the force $\mb{F}$ applied to it. For the fluid situation, this is written as
\begin{align}
\frac{\diff \mb{P}}{\diff t} := \frac{\diff}{\diff t}\oint_{c(\mb{u}_t^L)} \mb{u}_t(\mb x)\cdot {\sf d} \mb x
= \oint_{c(\mbs{u}_t^L)} \mb{f}(\mb x)\cdot {\sf d} \mb x =: \mb{F}
\,.
\label{Kel-forcelaw}
\end{align}
The Kelvin-Newton relation in \eqref{Kel-forcelaw} for loop momentum dynamics apparently involves two kinds of velocity. The first velocity is $\mb{u}_t^L$, which is the velocity of the material masses distributed in the line elements along the moving loop. Since it refers to the fluid parcel transport, the velocity $\mb{u}_t^L$ is a Lagrangian quantity. A second quantity with dimensions of velocity $(\mb{u}_t)$ appears in the integrand of the Kelvin circulation. This quantity is physically the momentum per unit mass, defined in the fixed inertial frame which is required for Newton's force law \eqref{Kel-forcelaw} to be valid. This means that $\mb{u}_t$ is an Eulerian quantity, defined in the fixed frame through which the Lagrangian parcels move at velocity $\mb{u}_t^L$. Mathematically, the momentum per unit mass $(\mb{u}_t)$ is the product of the  inverse of the mass density (which itself is a subset of the advected quantities, $D\subset a$) times the variational derivative at fixed spatial coordinate of the Lagrangian $\ell(\mb{u}_t^L,a)$ in Hamilton's principle with respect to the velocity, $\mb{u}_t^L$. In Euler--Poincar\'e form, this is the Kelvin--Noether theorem of \cite{HMR1998}. Namely, 
\begin{align}
\frac{\diff \mb{P}}{\diff t} 
:= \frac{\diff}{\diff t}\oint_{c(\mb{u}_t^L)}
\frac{1}{D}\frac{\delta \ell(\mbs{u}_t^L,a)}{\delta \mb{u}_t^L}\cdot {\sf d} \mb x
= \oint_{c(\mbs{u}_t^L)} \frac{1}{D} \frac{\delta \ell}{\delta a}\diamond a \cdot {\sf d} \mb x =: \mb{F}
\,,
\label{KelNoether-forcelaw}
\end{align}
where the diamond operation $(\diamond)$ is defined in \cite{HMR1998} and is discussed further in the present context below. 

Note, in the discussion below, when the Lagrangian velocity happens to be equal to the momentum per unit mass, then $\mb{u}_t^L\to \mb{u}_t$ and we shall drop the superscript $L$, although the distinction in their definitions still remains. This slight abuse of notation should cause no confusion, because the transport velocity is a vector field which acts on the momentum per unit mass which, in turn, is the 1-form appearing in the integrand of the Kelvin circulation integral. 
\bigskip

The modelling approach of Stochastic Advection by Lie Transport (SALT) modifies the Kelvin theorem in \eqref{Kel-forcelaw} for deterministic fluids by replacing the transport velocity of the loop $\mb{u}_t^L$ in the deterministic Kelvin theorem by a Stratonovich stochastic vector field ${\rm d} x_t$ whose drift velocity is the same as the Eulerian velocity in the \emph{integrand} of the deterministic Kelvin theorem  \cite{Holm2015},
\begin{align}\label{SALT-Kel}
\oint_{c(\mb u^L_t)} \mb u_t\cdot {\sf d} \mb x
\quad\to\quad 
\oint_{c(\diff x_t)}
\mb u_t\cdot {\sf d} \mb x\,,
\end{align}
where $\diff x_t$ denotes the following stochastic process,
\begin{align}\label{dx-form}
\diff {x_t} := \u^L_t (x_t)\diff t+ \displaystyle\sum_k \mb{\xi}_k (x_t)\circ \diff W_t
\,.
\end{align}

The vector fields $\mb{\xi}_k$ are to be determined from data analysis as in  \cite{CCHOS18b,CCHOS18a}. This paper will work formally, by simply assuming that these vector fields are already known from appropriate data analysis for a given application.

\begin{remark}[Notation temporal ($\diff$\,) vs spatial (${\sf d}$)]
In the literature, the letter $d$ is typically used to denote either (1) stochastic time evolution, or (2) exterior derivative/spacial differential.
To avoid confusion, here we will use the roman font $``\diff"$ to denote the former and the sans serif $``\sf{d}"$ to denote the latter.
\end{remark}

The same stochastic transport velocity ${\diff } x_t$ advects the Lagrangian parcels, which may carry advected quantities $(a)$, such as heat, mass and magnetic field lines, by Lie transport along with the flow, as
 \cite{HMR1998}
 \begin{align}\label{advec-qty}
\diff  {a} + \mathcal{L}_{\diff  {x_t} }a = 0\,.
\end{align}

In this paper, we apply the LA SALT (Lagrangian-averaged SALT) approach proposed in \cite{DrivasHolm2019} and developed in \cite{DHL2019}. The LA SALT approach modifies the SALT Kelvin circulation in \eqref{SALT-Kel} by replacing the drift velocity in the stochastic transport loop velocity in \eqref{dx-form} by its expectation, plus the same noise as in SALT. Namely, cf. equation \eqref{dx-form},
\begin{align}\label{KelThm-form}
\oint_{c({\diff} x_t)} \mb{u}_t  \cdot {\sf d} \mb x
\quad\to \quad
\oint_{c({\diff} X_t )} \mb{u}_t \cdot {\sf d} \mb x\,,
\end{align}
where 
\begin{align}\label{dX-form}
 \diff  X_t := \E{\u^L_t}(X_t) \diff t+ \displaystyle\sum_k {\mb \xi_k (X_t)} \circ \diff W_t
\,.
\end{align}
Since the expectation in \eqref{dX-form} refers to the transport velocity $u^L_t$ of Lagrangian loop in Kelvin's theorem, we refer to this process as probabilistic Lagrangian Average (denoted as LA), reminiscent of the time average at fixed Lagrangian coordinate in the LANS-alpha turbulence model,\cite{Chen-etal1998,Chen-etal1999,Chen-etal1999+,Foiasetal01,Foiasetal02}.
For example, in the Euler fluid case the modified Kelvin theorem reads,
\begin{align}\label{KelThm-Eul}
\diff  \oint_{c\big({\diff} X_t \big)} \mb{u}_t  \cdot {\sf d} \mb x
 =  
\oint_{c\big({\diff} X_t \big)}
\big[ \diff{\mb{u}_t} \cdot {\sf d} \mb x+  \mathcal{L}_{\diff X_t} u_t \big]
= 0 \,,
\end{align}
where $ \mathcal{L}_{ {\diff}X_t}u_t $ denotes the Lie derivative of the one-form 
$u_t = \mb{u}_t \cdot \sf d \mb{x}$ with respect to the vector field $ \diff X_t$ given in equation \eqref{dX-form}. The LA SALT motion equation leading to the modified Kelvin theorem in \eqref{KelThm-Eul} was previously stated along with additional noisy and viscous terms in Lemma 3 of \cite{DrivasHolm2019}. 

%%%%%%%%%%%%%%%%%%%%%%%%%%%%%%%%%%%%%%%%%%%

In fact, an alternative approach leading to the appearance of the vector field \eqref{dX-form} in a stochastic modification of the SALT Kelvin circulation theorem as in equation \eqref{KelThm-form} and leading to equation \eqref{KelThm-Eul} has also been proposed independently in \cite{Hoch2018}. In \cite{Hoch2018}, this modification was proposed as an analogue for SPDE of the McKean-Vlasov mean field approach for finite dimensional SDE describing Hamiltonian interacting particle systems when the Hamiltonian is independent of the position variables \cite{McKean}. The modification as in equation \eqref{KelThm-form} was applied in \cite{Hoch2018} to derive the Navier-Stokes equations by taking the expectation of the resulting equations.  

The present work will take the work in \cite{DrivasHolm2019} and \cite{Hoch2018} farther, by following the LA SALT (Lagrangian Averaged SALT) approach along the same lines as \cite{DHL2019} in applying expectations of the variations with respect to advected variables in combination with the known semidirect-product structure of the Lie--Poisson Hamiltonian formulation of ideal fluid dynamics. The semidirect-product structure of ideal fluid dynamics is reviewed for example \cite{MR2013, HSS2009}. 

To express the LA SALT equations discussed in \cite{DHL2019}, one may act with the semidirect-product (SDP) Lie--Poisson Hamiltonian matrix operator on the expected values of the variational derivatives of the Hamiltonian. In the absence of advected fluid quantities, the corresponding expected-quantity equations produce a Lie-Laplacian version of the Navier-Stokes equation, which reduces to the Navier--Stokes equation in a special choice of the functions $ \xi^{(k)}=\{(1,0,0)^T,(0,1,0)^T,(0,0,1)^T\}$ for $k=1,2,3$, as discussed in \cite{Hoch2018}. After writing the expectation equations with advected quantities in the SDP Hamiltonian matrix form, one observes that the fluctuation equations comprise a linear transport system which is slaved to the expectation equations whose solutions are deterministic and can be obtained for all time for a certain class of Hamiltonians. This slaving relation enables one to calculate the evolution equations for the local and spatially integrated variances of the fluctuations. This entire process will be pursued in this paper specifically for the LA SALT modification of the two-dimensional Euler--Boussinesq equations for a stratified incompressible fluid in a vertical plane.

%%%%%%%%%%%%%%%%%%%%%%%%%%%%%%%%%%%%%%%%

\bigskip

\subsection{The LA SALT 2D Euler--Boussinesq equations}
The deterministic Euler--Boussinesq (EB) equations for an incompressible, inviscid 2D fluid flow in a vertical plane under gravity are given by
\begin{equation}\label{Inviscid_Boussinesq}
\left\{
\begin{array}{rl}
\partial_t\u+\left(\u\cdot\nabla\right)\u &=-\nabla p+ g\theta \hat{\y}, \qquad (\x,t)\in \mathbb{T}^2\times\RR^{+},\\
\partial_t\theta +\u\cdot\nabla \theta &= 0, \\
\nabla\cdot\u&=0,
\end{array}
\right.
\end{equation}
where $\u=(u_{1},u_{2})$ is the incompressible vector velocity field, $p$ is the scalar pressure, $g$ is the acceleration due to gravity, $\theta$  corresponds to the temperature, or buoyancy, which is transported by the fluid, and $\hat{\y}$ is the unitary vector field in the vertical direction. The EB equations \eqref{Inviscid_Boussinesq} are fundamental in meteorology. Among other aspects, these equations are used to model the process of front formation. They are considered a fundamental model for the study of large scale atmospheric and oceanic flows, built environment, and dispersion of dense gases \cite{Ped87,Ric07}.  
%The system can be derived from density dependent fluid equations, using the so called \textit{Boussinesq approximation} which roughly speaking consists in neglecting the density dependence on all the terms but the one involving the gravity, [reference]. 
From a mathematical point of view, the 2D EB equations retain some key features of the well-known Euler and Navier-Stokes equations, as for instance, a vortex stretching mechanism for $\nabla\theta\times\hat{\y}\ne0$. The problem has attracted considerable attention in the PDE community, and local existence results and regularity criteria, as well as numerical experiments, are available, \cite{CanBen,HouLi1,Chae,ElgJeo}. The fundamental issue of whether classical solutions of the 2D incompressible Boussinesq equations can develop finite time singularities remains an outstanding open problem which seems to be out of reach, \cite{Yudo}.
In this paper, we will be dealing with the following LA SALT modification of the deterministic EB system in \eqref{Inviscid_Boussinesq} as a suitable model for predicting EB `climate' dynamics in the sense of Lorenz \cite{Lorenz1995}.

\begin{equation}\label{LA:SALT:Lie:Bou}
\left\{
\begin{array}{rl}
\diff {u}\ +& \mathcal L_{\mathbb E[u]} {u} \diff t + \displaystyle \displaystyle\sum_k \mathcal L_{\xi_k} {u} \circ \diff W_t^k 
\\=& - \,{\sf{d}} \E{ p -  |\u|^2/2 }\, \diff t + g \E{\theta} \mb{\hat{y}} \diff t
\,- gy {\sf{d}} (\theta - \E{\theta}) \diff t
\,,\\ \\
\diff \theta\ +& \mathcal L_{\mathbb E[u]} {\theta} \diff t + \displaystyle \displaystyle\sum_k \mathcal L_{\xi_k} {\theta} \circ \diff W_t^k = 0
\,,\hfill \nabla \cdot \mathbb{E} \left[ \u \right]=0.
\end{array}
\right.
\end{equation}
\begin{remark}[Divergence-free condition on the expectation of the velocity] We note that although the more restrictive divergence-free condition $\nabla \cdot \u = 0$ might seem more natural to consider at first sight than our current condition $\nabla \cdot \mathbb{E} \left[ \u \right]$, it would make equations \eqref{LA:SALT:Lie:Bou} ill-posed. This is due to the presence of the term $\nabla \E{ p -  |\u|^2/2 },$ which imposes the pressure to be deterministic. Further insight into this will be provided once we present our approach for solving equations \eqref{LA:SALT:Lie:Bou}. Here, we simply note that if the expectation in the term $\nabla \E{ p -  |\u|^2/2 }$ is removed, the condition $\nabla \cdot \u = 0$ could be considered.

\end{remark}
In the equations above, we have employed the notation $\mathcal{L}_{\xi_k}$ to indicate Lie derivative along a vector field. As stressed in Subsection \ref{2-1}, the Lie derivative on one-forms 
\[
\mathcal{L}_{\xi} u = \mb{\xi} \cdot \nabla \u + \displaystyle\sum_{j} \u^j \nabla {\mb \xi}^j
\]
is different from the Lie derivative applied to scalar fields $\mathcal{L}_{\xi} \theta = \mb{\xi} \cdot \nabla \theta.$ 

As explained below in Example \ref{EB-Ham-example2.1} the system \eqref{LA:SALT:Lie:Bou} 
can be rewritten in Hamiltonian operator form as 
\begin{align}
\diff 
\begin{bmatrix}
\mu \\ \\  \theta \\ \\ \rho 
\end{bmatrix}
= -
\begin{bmatrix}
 \mathcal{L}_{\Box} \mu &  - \Box  (\nabla \theta)    &  \rho \nabla {\Box}
\\ \\
\Box \cdot (\nabla \theta) & 0 & 0 \\ \\
 \nabla \cdot  (\rho \Box) & 0 & 0
\end{bmatrix}
\begin{bmatrix}
 \E{u}\diff t + \sum_k {\mb \xi_k}\circ \diff W^{(k)}_t
\\ \\ -g y \E{\rho} \diff t \\ \\ \E{ p  - \frac{ |\mb{u}|^2}{2} }\diff t  - g\E{\theta} y \diff t
\end{bmatrix},
\label{Ham-matrix-Boussinesq2-1}
\end{align}
which yields equations \eqref{LA:SALT:Lie:Bou}. 
Upon passing to the It\^o formulation, the LA SALT EB system \eqref{LA:SALT:Lie:Bou} transforms into 
\begin{equation}\label{LA:SALT:Ito:Bou}
\left\{
\begin{array}{rl}
\diff {u} + \mathcal L_{\mathbb E[u]} {u} \diff t 
+  \displaystyle\sum_k \mathcal L_{\xi_k} {u} \, \diff W_t^k 
&= - {\sf{d}} \E{p -  |\mb{u}|^2/2} \, \diff t + g \E{\theta} \mb{\hat{y}} \diff t 
\\ & \quad
- gy {\sf{d}} (\theta - \E{\theta}) \diff t 
+ \displaystyle \frac12 \sum_k \mathcal L_{\xi_k}^2 {u} \diff t, 
\\
\diff \theta + \mathcal L_{\mathbb E[u]} {\theta} \diff t + \displaystyle \displaystyle\sum_k \mathcal L_{\xi_k} {\theta}\diff W_t^k &=  \displaystyle \displaystyle \frac12 \sum_k \mathcal L_{\xi_k}^2 {\theta} \diff t\,,
\end{array}
\right.
\end{equation}
where we denote the composition of Lie derivatives as, for example,  $\mathcal L_{\xi_k}(\mathcal L_{\xi_k} {\theta}) =: \mathcal L_{\xi_k}^2 {\theta}$.

Next, taking expectation at both sides of the equations above yields a deterministic equation for the evolution of the expectations given by

\begin{equation}\label{LA:SALT:Exp:Bou}
\left\{
\begin{array}{rl}
\partial_t \mathbb E[u] + \mathcal L_{\mathbb E[u]} \mathbb E[u] &= -{\sf{d}} \left(\mathbb E[p] - \mathbb E \left[|\u|^2/2\right]\right) + g \mathbb E[\theta] \mb{\hat{y}} +  \displaystyle \displaystyle \frac12 \sum_k \mathcal L_{\xi_k}^2 \mathbb E[u], \\
\partial_t \mathbb E[\theta] + \mathcal L_{\mathbb E[u]} \mathbb E[\theta] &= \displaystyle \displaystyle \frac12 \sum_k \mathcal L_{\xi_k}^2 \mathbb E[\theta].
\end{array}
\right.
\end{equation}
It is straightforward to check that in vorticity form where $\omega=\nabla^{\perp}\cdot \mb u = \hat{\y}\cdot{\rm curl} \mb u$, we have that 
\begin{equation}\label{LA:SALT:Stra:Bou}
\left\{
\begin{array}{rl}
\diff \omega + \mathcal L_{\mathbb E[u]} \omega \diff t + \displaystyle \displaystyle\sum_k \mathcal L_{\xi_k} \omega \,\circ \diff W_t^k &= g \p_{x}\theta \diff t, \\
\diff \theta + \mathcal L_{\mathbb E[u]} {\theta} \diff t + \displaystyle \displaystyle\sum_k \mathcal L_{\xi_k} {\theta} \circ \diff W_t^k &=0.
\end{array}
\right.
\end{equation}
We stress here again that since $\omega$ is a scalar quantity for incompressible planar flow, its Lie derivative is to be understood as $\mathcal{L}_{\xi} \omega = \xi \cdot \nabla \omega.$ The corresponding equation for the expectation is given by
\begin{equation}\label{LA:Vor:Exp:Bou}
\left\{
\begin{array}{rl}
\partial_t \mathbb E[\omega] + \mathcal L_{\mathbb E[u]} \mathbb E[\omega] &= g \mathbb \p_{x} \E\theta+  \displaystyle \frac12 \displaystyle\sum_k \mathcal L_{\xi_k}^2 \mathbb E[\omega], \\
\partial_t \mathbb E[\theta] + \mathcal L_{\mathbb E[u]} \mathbb E[\theta] &=  \displaystyle \displaystyle \frac{1}{2} \sum_k \mathcal L_{\xi_k}^2 \mathbb E[\theta].
\end{array}
\right.
\end{equation}

\section{Lagrangian-averaged (LA) semidirect product systems with transport noise}\label{sec:LASDPsystems}
In subsequent discussions, we will employ the following notations:
\begin{itemize}
\item $M$ is a smooth, orientable manifold,
\item $\rm{Diff}(M)$ denotes the group of diffeomorphsims on $M$,
\item $\mathfrak X(M)$ denotes the set of smooth vector fields on $M$,
\item $\Omega^1(M)$ denotes the set of differential one-forms on $M$,
\item $\rm{Den}(M)$ denotes the set of volume forms (densities) on $M$,
\item $V$ is any tensor field such that $\rm{Diff}(M)$ acts on it from the right (e.g. $V = C^\infty(M,\mathbb R)$ and $\rm{Diff}(M)$ acts on $V$ by composition from the right).
\end{itemize}

\subsection{Poisson structure of fluid equations with advected quantities.} \label{2-1}

We have introduced a class of stochastic partial differential equations (SPDE) for continuum dynamics. This class of equations is Hamiltonian with a Lie--Poisson bracket given by the $L^2$ pairing between $\mathfrak{X}(M) \circledS V$ and its dual \cite{HMR1998} 
\begin{align}
\frac{\diff F}{\diff t}=\{ F, H\} 
= 
-\,\SCP{(\mu,a)}{\left[ \frac{\delta F}{\delta (\mu,a)}\,,\, \frac{\delta H}{\delta (\mu,a)} \right]}_{\mathfrak{X}(M)\circledS V},
\label{LP-Brkt}
\end{align}
where $F,H\in C^\infty(\mathfrak X^*(M) \times V^* \to \mathbb R)$, $\mu \in \mathfrak X^*(M) \cong \Omega^1(M) \otimes \rm{Den}(M)$, $a\in V^*$, ${\delta F}/{\delta (\mu,a)}\in \mathfrak{X}(M) \circledS V$ is the variational derivative (see \cite{marsden1983coadjoint}), and $\mathfrak{X}(M) \circledS V$ denotes the semidirect product Lie algebra of vector fields on $M$ acting on the vector space $V$. The square brackets $[\,\cdot\,,\,\cdot\,]$ denote the adjoint action of the semidirect product Lie algebra $\mathfrak{X}(M) \circledS V$ on itself.

Upon integration by parts, the Lie--Poisson bracket in \eqref{LP-Brkt} may be expressed in terms of a Hamiltonian operator as
\begin{align}
\frac{\diff F}{\diff t}=\{ F, H\}
=
-\int_M 
\begin{bmatrix}
{\delta F}/{\delta \mu} \\  {\delta F}/{ \delta a} 
\end{bmatrix}^T
\begin{bmatrix}
{\rm ad}^*_{\Box}\mu & \Box \diamond a
\\ 
\mathcal{L}_{\Box}a & 0
\end{bmatrix}
\begin{bmatrix}
{\delta H}/{\delta \mu} \\  {\delta H}/{ \delta a}  
\end{bmatrix}
{\sf d}V
\label{Ham-matrix-det}
\end{align}
where ${\rm ad^*} : \mathfrak X(M) \times \mathfrak X^*(M) \rightarrow \mathfrak X^*(M)$ is the coadjoint action, $\mathcal L_u \alpha$ is the Lie derivative of a tensor field $\alpha$ with respect to a vector field $u$, and the diamond operator $\diamond: V\times V^*\to \mathfrak{X}^*(M)$  is defined in terms of the Lie derivative as, 
\begin{align}
\Scp{b\diamond a}{v}_{\mathfrak{X}(M)} := \Scp{b }{- \mathcal{L}_v a}_V 
\,,\label{diamond-def}
\end{align}
where $a\in V^*$ and $b\in V$. 
The definition \eqref{diamond-def} makes the Lie--Poisson bracket skew-symmetric in $L^2$ under integration by parts.

We note that the Lie derivative $\mathcal L$ has different local expressions depending on which type of tensor field it acts on, which we will list below. Let $u \in \mathfrak X(\mathbb R^n)$ for all examples below.
\begin{itemize} 
    \item (Scalar functions) Given a scalar field $f$, we have
    $$\mathcal L_u f = \mb{u} \cdot \nabla f.$$
    \item (Vector fields) If $v \in \mathfrak X(\mathbb R^n)$ is another vector field, then
    $$\mathcal L_u v = \left(\mb{u}\cdot \nabla \mb{v} - \mb{v}\cdot \nabla \mb{u}\right)\cdot \nabla = [u,v] = -{\rm ad}_u v.$$
    \item (One-forms) Given a one-form $\alpha \in \Omega^1(\mathbb R^n)$, the corresponding Lie derivative reads
    $$\mathcal L_{u}\alpha = \left(\mb{u} \cdot \nabla \mb{\alpha} + \sum_{j=1}^n \alpha_j \nabla u^j\right) \cdot {\sf d} \x.$$
    \item (Densities) Given a density $D = \rho {\sf d} x^n \in \Omega^n(\RR^n)$, we have
    $$\mathcal L_{u} D = \div(\rho \u) {\sf d} x^n.$$
    \item (One-form densities) Given a one-form density $\mu = \alpha \otimes \rho\,{\sf d} x^n,$ where $\alpha \in \Omega^1(\mathbb R^n)$ and $\rho\,{\sf d} x^n \in \Omega^n(\RR^n)$, its Lie derivative is given by
    $$\mathcal{L}_u (\alpha \otimes \rho\,{\sf d} x^n) = (\rho\, \mathcal{L}_u \alpha + \div(\rho \u) \alpha) \otimes {\sf d} x^n. $$
    It is well-known that for one-form densities (which are dual under $L^2$ pairing to the Lie algebra of vector fields), the coadjoint representation of the Lie algebra is equivalent to the Lie derivative, i.e., ${\rm ad}^*_u (\alpha \otimes \rho\,{\sf d} x^n) \equiv \mathcal{L}_u (\alpha \otimes \rho\,{\sf d} x^n)$, a fact we will use throughout this paper.
\end{itemize}
We refer the readers to \cite{HMR1998} for further examples of Lie derivatives arising in continuum dynamics and the corresponding expressions for the diamond operator. We also remark that all the previous definitions take the same form on the torus $\mathbb{T}^2.$
\begin{example}[The deterministic 2D Euler-Boussinesq equations] \label{EB-Ham-example2.1}
We recall that the Boussinesq system is given by
\begin{equation}\label{Inviscid_Boussinesq00}
\left\{
\begin{array}{rl}
\partial_t\u+\left(\u\cdot\nabla\right)\u &=-\nabla p + g\theta \hat{\y}, \qquad (\x,t)\in \mathbb{T}^2\times\RR^{+},\\
\partial_t\theta +\u\cdot\nabla \theta &= 0, \\
\nabla\cdot\u&=0,
\end{array}
\right.
\end{equation}
where $\u=(u_{1},u_{2})$ is the incompressible vector velocity field, $p$ is the scalar pressure, $g$ is the acceleration due to gravity, and $\theta$  corresponds to the temperature, which is transported by the fluid. In Lie--Poisson form with $(\mu,\theta,D)$ denoting momentum one-form density, potential temperature, and density respectively, where $\mu (x,t) :=  \mb u \cdot {\sf d} \x \otimes \rho \, {\sf d} x^2,$ $D := \rho \, {\sf d} x^2,$ and the advected potential temperature $\theta = \theta(x,t)$ is understood as a scalar quantity. In the semidirect product formalism presented in \eqref{Ham-matrix-det}, this can be expressed as
\begin{align}
\diff{F} =\{ F, h\}
=
-\bigints_{\mbox{$\Large \mathbb{T}^2$}}
\begin{bmatrix}
{\delta F}/{\delta \mu} \\  {\delta F}/{ \delta \theta}  \\  {\delta F}/{ \delta D}
\end{bmatrix}^T
\begin{bmatrix}
{\rm ad}^*_{\Box} \mu & \Box \diamond \theta & \Box \diamond D
\\ 
\mathcal{L}_{\Box} \theta & 0 & 0 \\
\mathcal{L}_{\Box} D & 0 & 0
\end{bmatrix}
\begin{bmatrix}
{\delta H } / {\delta \mu}
\\  {\delta H}/{ \delta \theta} \\
 {\delta H}/{ \delta D} 
\end{bmatrix}
{\sf d} x^2,
\label{Ham-matrix-Boussinesq}
\end{align}
for Boussinesq Hamiltonian $h$ given in terms of $(\mu,\theta,D)$ by the sum of the kinetic and potential energies, plus a constraint applied by the Lagrange multiplier $p$ (the pressure) which enforces incompressibility
\begin{align}
\begin{split}
h(\mu,\theta,\rho) &=  \int_{\mathbb{T}^2} \left(\frac{1}{2\rho} |\mu|^2  
-  g\rho  \theta y + p (\rho-1) \right) \,{\sf d} x^2 
\\&= \int_{\mathbb{T}^2} \left \langle \mu, u \right \rangle 
- \int_{\mathbb{T}^2} \left(\frac{\rho}{2} |\mb{u}|^2 + g\rho \theta y - p (\rho-1) \right) \,{\sf d} x^2
\,,
\end{split}
\label{Boussinesq-Ham}
\end{align}
so that
\begin{align}
\frac{\delta h}{\delta \mu} = u := \mb u\cdot\nabla
\,,\quad
\frac{\delta h}{ \delta u} = \mu - \rho u = 0
\,,\quad
\frac{\delta h}{ \delta \theta} = -g\rho y
\,,\quad
\frac{\delta h}{\delta \rho} = p - \frac{|\mu|^2}{2\rho^2}- g\theta y.
\label{EBHam-var}
\end{align}
We note that the constraint coming from the Lagrangian multiplier $p$ yielding $\rho = 1$ is only to be imposed once the variations are taken and the final equations derived.
The definitions for the Lie-derivative, diamond, and coadjoint operator ${\rm ad}^*$ have been specified above. We note that these depend on the type of object they are being applied to (i.e. $\mu$ is a one-form density, whereas $\theta$ a scalar, and $D$ a volume form). Upon applying these definitions, we can rewrite \eqref{Ham-matrix-Boussinesq} as 
\begin{align}
\partial_t
\begin{bmatrix}
\mu \\  \theta \\ \rho 
\end{bmatrix}
= -
\begin{bmatrix}
 \mathcal{L}_{\Box} \mu &  - \Box  (\nabla \theta)    &  \rho \nabla {\Box}
\\ 
\Box \cdot (\nabla \theta) & 0 & 0 \\
 \nabla \cdot  (\rho \Box) & 0 & 0
\end{bmatrix}
\begin{bmatrix}
 u
\\  -g\rho y \\ p - |\mb{u}|^2/2 - g\theta y
\end{bmatrix},
\label{Ham-matrix-Boussinesq2}
\end{align}
which yields equations \eqref{Inviscid_Boussinesq00}. 
\end{example}

\subsection{SALT equations.}

The class of Hamiltonian SPDE treated here may be obtained by extending the Hamiltonian function to make it stochastic by adding the $L^2$ pairing of the momentum density $\mu$ with a Stratonovich stochastic process (denoted with the symbol $\circ \diff W_t$) whose spatial correlations are specified by a set of smooth vector fields, $\mb{\xi}_k(\x)$, $k=1,\dots,N$, as in 
\cite{Holm2015}, as
\begin{align}
H(\mu,a) \to \diff {h(\mu,a;\xi_k) }
:= H(\mu,a)\diff t + \displaystyle\sum_k\scp{\mu}{\xi_k}\circ \diff W_t^k
\,.\label{Stoch-Ham}
\end{align}
The Lie--Poisson bracket then yields
\begin{align}
\diff  {F} =\{ F, \diff {h}\}
=
-\int_M 
\begin{bmatrix}
{\delta F}/{\delta \mu} \\  {\delta F}/{ \delta a} 
\end{bmatrix}^T
\begin{bmatrix}
{\rm ad}^*_{\Box}\mu & \Box \diamond a
\\ 
\mathcal{L}_{\Box}a & 0
\end{bmatrix}
\begin{bmatrix}
({\delta H } / {\delta \mu})\diff t + \displaystyle\sum_k\mb{\xi}_k(\x)\circ \diff W_t^k
\\  ({\delta H}/{ \delta a} ) \diff t 
\end{bmatrix}
{\sf d}V.
\label{Ham-matrix-SALT}
\end{align}
These equations describe stochastic advection by Lie transport (SALT)
\cite{Holm2015} and they comprise the basis for a new approach for data analysis, uncertainty quantification and uncertainty reduction by data assimilation using particle filtering \cite{CCHOS18b,CCHOS18a}. By defining the stochastic vector field 
\begin{align}
\diff{x_t} := 
({\delta H } / {\delta \mu})\diff t + \displaystyle\sum_k \mb{\xi}_k(\x)\circ \diff W_t^k
\label{Stoch-Lag-traj}
\end{align}
and recalling that ${\rm ad}^*_{{\rm d}x_t}\mu
= \mathcal{L}_{{\rm d}x_t}\mu,$ 
the SALT equations \eqref{Ham-matrix-SALT} may be rewritten in a compact form as 
\begin{align}
\begin{split}
\diff{\mu} + \mathcal{L}_{\diff{x_t}}\mu 
&= - \frac{\delta H }{\delta a}\diamond a \,\diff t
\,,\\
\diff{a} + \mathcal{L}_{\diff{x_t}}a &= 0
\,.
\end{split}
\label{SALT-adv-form}
\end{align}
The SALT equations in this form have been studied extensively, for example, in 
wave-current interactions \cite{Holm_WCI_JNLS2019}, uncertainty prediction \cite{GH18}, solution properties of stochastic fluid dynamics  \cite{CrFlHo2019, AloBetTak}, and turbulent cascades \cite{HolmTurbulent}, even when the spatial correlations are nonstationary \cite{GH18a}. 

\begin{example}[SALT 2D Euler-Boussinesq system]
The 2D SALT Boussinesq equations are given by
\begin{align}\label{Ham-matrix-Boussinesq-SALT}
\diff  {F} =\{ F, h\}
=
-\bigints_{\mbox{$\Large \mathbb{T}^2$}}
\begin{bmatrix}
{\delta F}/{\delta \mu} \\  {\delta F}/{ \delta \theta}  \\  {\delta F}/{ \delta D}
\end{bmatrix}^T
\begin{bmatrix}
{\rm ad}^*_{\Box} \mu & \Box \diamond \theta & \Box \diamond D
\\ 
\mathcal{L}_{\Box} \theta & 0 & 0 \\
\mathcal{L}_{\Box} D & 0 & 0
\end{bmatrix}
\begin{bmatrix}
{\delta h } / {\delta \mu}
\\  {\delta h}/{ \delta \theta} \\
 {\delta h}/{ \delta D} 
\end{bmatrix}
{\sf d} x^2,
\end{align}
\end{example}
with
\begin{align}\label{Boussinesq-Ham-SALT}
h(\mu,\theta,D) = \int_{0}^t \int_{\mathbb{T}^2} \left(\frac{1}{2\rho} |\mu|^2  +  g \rho \theta y + p (\rho-1) \right) \,{\sf d} x^2 \diff s +  \displaystyle\sum_k\int_{0}^t \int_{\mathbb{T}^2} \left \langle \mu (x,t), \xi_k \right \rangle \, \circ \diff W_s^k
\,,
\end{align}
where $\mu = \rho \mb{u}  \cdot \sf{d} \mb{x} \otimes$ $\diff^2 x$ and $D = \rho \diff^2 x$ giving rise to the SALT 2D Euler--Boussinesq (EB) system
\begin{equation}\label{Inviscid_Boussinesq00-SALT}
\left\{
\begin{array}{rl}   
\diff \u + \u \cdot \nabla \u \diff t + \displaystyle\sum_k \mb{\xi_k} \cdot \nabla u \circ \diff W_t^k + \displaystyle\sum_k u^j \nabla \xi_k^j \circ \diff W_t^k &= -{\sf{d}}(p- |\u|^2/2) \diff t + g\theta \hat{\y} \diff t,\\
\diff \theta + \u \cdot \nabla \theta \diff t + \displaystyle\sum_k \mb{\xi_k} \cdot \nabla \theta \circ \diff W_t^k &= 0, \\
\nabla\cdot\u&=0.
\end{array}
\right.
\end{equation}

We note that the well-posedness of this equation and a blow-up criterion for it were derived in \cite{DieAytJNLS}. In this paper, by considering the Lagrangian-averaged version of \eqref{Inviscid_Boussinesq00-SALT}, we construct the LA SALT 2D EB model, which will turn out to be \emph{globally} well-posed.

\subsection{Lagrangian-averaged (LA) SALT equations.}

Recently a modification of the SALT has been made in \cite{DrivasHolm2019} and analysed in \cite{DHL2019} for 3D stochastic fluid motion. This modification preserves the Lie--Poisson bracket structure of the SALT equations, while replacing the variational derivatives of the Hamiltonian by their expected values, denoted $\mathbb{E}[\,\cdot\,]$, as follows. First, the Lagrangian trajectory equation \eqref{Stoch-Lag-traj} is modified by taking the expectation of the drift velocity, as
\begin{align}
\diff{X_t} := 
\E{\frac{\delta H }{\delta \mu}} \diff t + \displaystyle\sum_k\xi_k(x)\circ \diff W_t^k,
\label{Exp-Lag-traj}
\end{align}
where $H$ is the same Hamiltonian as in the SALT equations. We also take the expectation of the variational derivatives with respect to advected quantities $\E{{\delta H }/{\delta a}}$.

The Poisson operator then yields 
\begin{align}
\diff  {F} =\{ F, \diff {h}\}
=
-\int_M 
\begin{bmatrix}
{\delta F}/{\delta \mu} \\  {\delta F}/{ \delta a} 
\end{bmatrix}^T
\begin{bmatrix}
{\rm ad}^*_{\Box}\mu & \Box \diamond a
\\ 
\mathcal{L}_{\Box}a & 0
\end{bmatrix}
\begin{bmatrix}
\mathbb{E}\left[\frac{\delta H }{\delta \mu}\right] \diff t + \displaystyle\sum_k\xi_k(x)\circ \diff W_t^k
\\  \mathbb{E}\left[\frac{\delta H }{\delta a} \right] \diff t 
\end{bmatrix}
{\sf d}V \,.
\label{Ham-matrix-LASALT}
\end{align}
These equations describe \emph{Lagrangian-averaged} stochastic advection by Lie transport (LA SALT). That is, the Lagrangian path ${\rm d}X_t$ in equation \eqref{Exp-Lag-traj} has been acquired by taking the expectation (averaging in probability space) of the drift velocity of the SALT Lagrangian path \eqref{Stoch-Lag-traj} at \emph{fixed Lagrangian label}.
The SALT equations in advective form \eqref{SALT-adv-form} now become the LA SALT equations, given by
\begin{equation}\label{LASALT-adv-form}
\left\{
\begin{array}{rl}
{\rm d}{\mu} + \LL_{\E{\frac{\delta H }{\delta \mu}}} {\mu} \,\diff t
+ \displaystyle\sum_k \LL_{\xi_k} {\mu} \circ \diff W_t^k
&= - \,\mathbb{E}\Big[ \frac{\delta H}{\delta a}\Big] \diamond { a }\,\diff t
\,,\\
 {\rm d} {a} + \LL_{\E{\frac{\delta H }{\delta \mu}}} {a} \,\diff t
 + \displaystyle\sum_k \LL_{\xi_k} {a} \circ \diff W_t^k  &= 0
\,,
\end{array}
\right.
\end{equation}
with ${\rm d}X_t$ defined in equation \eqref{Exp-Lag-traj}. If there are several advected quantities, one sums over all of them in the diamond term in \eqref{LASALT-adv-form}. Notice that the LA SALT equations in \eqref{Ham-matrix-LASALT} have the same Poisson matrix operator as for the SALT equations in \eqref{Ham-matrix-SALT} and therefore many key features of the Lie-Poisson system are preserved, such as the conservation of Casimirs and Kelvin's circulation theorem (see Remark \ref{kelvin} below).
Thus, between equations \eqref{Ham-matrix-SALT} and \eqref{Ham-matrix-LASALT}, only the variational derivatives of the deterministic parts of the Hamiltonian have been changed to accommodate the differences between Lagrangian trajectories for SALT and LA SALT in equations \eqref{Stoch-Lag-traj} and \eqref{Exp-Lag-traj}.

\begin{remark}[Comparing SALT and LA SALT]\rm
The LA SALT approach applies to the same physical class of equations as for SALT. Following the deterministic route set in \cite{HMR1998}, the class of SALT fluid equations was first derived in \cite{Holm2015} from the symmetry-reduced Lagrangians $\ell(u,a)$ for the Euler--Poincar\'e Hamilton's principle with $\mu=\delta \ell/\delta u$, whose variations were constrained to respect stochastic advection laws in \eqref{SALT-adv-form}. The LA SALT approach modifies the stochastic process $ \diff{x_t}$ for the transport vector field in \eqref{Stoch-Lag-traj} which defines the stochastic Lagrangian trajectory in SALT to become $ \diff{X_t}$ as in \eqref{Exp-Lag-traj}. The Euler--Poincar\'e version of the Lie--Poisson expression of the motion equation in \eqref{LASALT-adv-form} is, 
\begin{align}\label{LASALT-EP}
 \diff{\frac{\delta \ell}{\delta u}} + \mathcal{L}_{ \diff{X_t}} \frac{\delta \ell}{\delta u}  
 = \mathbb{E}\Big[ \frac{\delta \ell}{\delta a}\Big] \diamond a\,\diff t
 \quad\hbox{and}\quad
  {\rm d}a + \mathcal{L}_{ {\rm d}X_t} a = 0
\,.
\end{align}
The comparisons between them can be derived from the relations ${\delta \ell}/{\delta u}=\mu$ and ${\delta \ell}/{\delta a}=-{\delta h}/{\delta a}$ which are obtained from the deterministic Legendre transform from the reduced Lagrangian to the reduced Hamiltonian,
\begin{align}
\diff h(\mu,a) = \scp{\mu}{u} - \ell(u,a)
\,,
\label{Legendre-xform}
\end{align}
and the assumption that the reduced Lagrangian is hyperregular, which almost always holds in continuum mechanics. \hfill $\square$
\end{remark}

\begin{remark}[The Kelvin circulation theorem for LA SALT]\label{kelvin} \rm
In fluid dynamics, the mass density $D {\sf d}^3x$ is always an advected quantity, satisfying 
the continuity equation, which in this case is expressed as,
\begin{align}
\diff {(D {\sf d}^3x)} + \mathcal{L}_{\diff{X_t}}(D {\sf d}^3x)  
= \big(\diff{D} + {\rm div}(\diff{X_t} D)\big) {\sf d}^3x
= 0
\,.
\label{Contin-eqn}
\end{align}
Consequently, if we define the circulation one-form $v= \mb{v}\cdot {\sf d} \mb{x}$ by 
\begin{align}
\mu = \mb{m}\cdot {\sf d} \mb{x} \otimes {\sf d}^3x = \mb{v}\cdot {\sf d} \mb{x} \otimes D {\sf d}^3x
= {v} \otimes D {\sf d}^3x
\,,\label{Circ-1form}
\end{align}
and use the continuity equation \eqref{Contin-eqn},
and then the advective form of the motion equation in \eqref{LASALT-adv-form}, we can write the Kelvin circulation theorem for LA SALT as
\begin{align}
\rmd \oint_{c(\diff{X_t})} \!\!\!{\mb v}\cdot {\sf d} {\mb x}  
= \oint_{c(\diff{X_t})} \!\!\!
\big(\diff +  \mathcal{L}_{\diff{X_t}}\big)({\mb v}\cdot {\sf d}{\mb x})  
= - \oint_{c(\diff{X_t})} \frac{1}{D} 
\mathbb{E}\Big[\frac{\delta H }{\delta a}\Big]\diamond a
\,.\label{Kel-thm-LASALT}
\end{align}
This relation may be proved, for example, by following the corresponding proof of the stochastic Kelvin calculation for SALT in  \cite{de2019implications}. Thus, because the LA SALT modification in \eqref{Exp-Lag-traj} of the SALT transport vector field in \eqref{Stoch-Lag-traj} preserves the Lie--Poisson Hamiltonian structure of SALT, one also acquires the Kelvin circulation theorem for LA SALT in \eqref{Kel-thm-LASALT}. Note that for compressible fluids, the right-hand side of the relation in \eqref{Kel-thm-LASALT} can be nonlinear in the stochastic variables. \hfill $\square$
\end{remark}

\subsection{It\^o solutions of LA SALT dynamics}
The solution behaviour in the It\^o version of LA SALT dynamics has stochastic Lagrangian paths given by \cite{Gardiner}
\begin{align}
\begin{split}
\diff{\widehat{X}_t }
&=
\mb{\widehat{U}}(\mb x,t) \diff t + \displaystyle\sum_k \mb{\xi}_k(\mb x) \diff W_t^k
\\&:= 
\bigg(\mathbb{E}\left[\frac{\delta H }{\delta \mu}\right]
+ \frac12 \displaystyle\sum_k \big(\mb{\xi}_k\cdot\nabla\big)\mb{\xi}_k(\mb x)\bigg) \diff t 
+ \displaystyle\sum_k \mb{\xi}_k(\mb x) \diff W_t^k\,.
\end{split}
\label{Exp-Lag-traj-ito}
\end{align}
and we can re-write equation \eqref{LASALT-adv-form} in It\^o form as
\begin{equation}\label{LASALT-EPlinIto}
\left\{
\begin{array}{rl}
{\rm d}{\mu} + \LL_{\E{\frac{\delta H }{\delta \mu}}} {\mu} \,\diff t
+ \displaystyle\displaystyle\sum_k \LL_{\xi_k} {\mu} \,\diff W_t^k
- \frac{1}{2}  \sum_k \LL_{\xi_k}( \LL_{\xi_k} {\mu}) \diff t 
&= - \,\mathbb{E}\Big[ \frac{\delta H}{\delta a}\Big] \diamond { a }\,\diff t
\,,\\
 {\rm d} {a} + \LL_{\E{\frac{\delta H }{\delta \mu}}} {a} \,\diff t
 + \displaystyle\sum_k \LL_{\xi_k} {a} \, \diff W_t^k
 - \frac{1}{2}  \displaystyle\sum_k \LL_{\xi_k}( \LL_{\xi_k} {a})\,\diff t  &= 0
\,.
\end{array}
\right.
\end{equation}

The  It\^o LA SALT dynamics turns out to be quite different from that of It\^o SALT dynamics. Indeed, fundamental and significant  simplifications occur in the structure of the equations when the drift velocity of SALT is replaced by its expectation in LA SALT.  First, when the expectations of the two LA SALT equations in advective form \eqref{LASALT-adv-form} are written out by taking the expectation on both sides of \eqref{LASALT-EPlinIto}, noting that the It\^o integral vanishes due to the martingale property,
\begin{equation}\label{LASALT-EPX}
\left\{
\begin{array}{rl}
\partial_t \E{\mu} + \LL_{\E{\frac{\delta H }{\delta \mu}}} \E{\mu} 
- \frac{1}{2} \displaystyle\sum_k \LL_{\xi_k}( \LL_{\xi_k} \E{\mu})  
&= - \,\mathbb{E}\Big[ \frac{\delta H}{\delta a}\Big] \diamond \E{ a }
\,,\\
 \partial_t \E{a} + \LL_{\E{\frac{\delta H }{\delta \mu}}} \E{a} - \frac{1}{2}  \displaystyle\sum_k \LL_{\xi_k}( \LL_{\xi_k} \E{a})  &= 0
\,,
\end{array}
\right.
\end{equation}
one realises that these equations provide the \emph{entire history} of the solutions for the expectations $\E{\mu}$ and $\E{a}$ throughout the domain of flow. Once the expectation equations \eqref{LASALT-EPX} have been solved, the equations for the instantaneous stochastic variables \eqref{LASALT-EPlinIto} become \emph{linear} It\^o stochastic transport equations, which are driven by the solutions of equations \eqref{LASALT-EPX}, whose entire history is obtained separately. We note that the coupled system \eqref{LASALT-EPlinIto}+\eqref{LASALT-EPX} is closed provided that the variables $(\mu,a)$ and its corresponding variational derivatives are related linearly, i.e., there exists a linear operator $\Lambda : \mathfrak{X}(M) \circledS V \rightarrow \mathfrak{X}^*(M) \circledS V^*$ such that $\left({\delta H}/{\delta \mu}, {\delta H}/{\delta a}\right) = \Lambda^* (\mu, a)$. For example, $\Lambda$ is a convolution with respect to some kernel.

\begin{remark}[Non-parabolicity of the It\^o equation]
We note that the presence of the second order differential operator $\mu\to- \frac{1}{2}  \sum_k \LL_{\xi_k} \LL_{\xi_k}\mu$ in the It\^o formulation \eqref{LASALT-EPlinIto} does not introduce parabolicity into the equation even though $-\frac{1}{2}  \sum_k \LL_{\xi_k}\LL_{\xi_k}\mu$ reduces to the standard Laplace operator $-\Delta \mu$ when $\xi^{(1)} = \hat{\mb{x}},$ $\xi^{(2)}=\hat{\mb{y}}$ . This feature of It\^o calculus may be understood and demonstrated as follows. First, the It\^o and Stratonovich formulations (\eqref{LASALT-EPlinIto} and \eqref{LASALT-adv-form}, respectively) are equivalent, and the latter is a pure transport equation. Second, an additional term appears in the process of making energy estimates in the It\^o formulation. This is known as the It\^o correction term, and it cancels the a priori  \textit{dissipative} effect of the double Lie derivative. Consequently, although one may expect to show that the initial smoothness of the equations will be preserved, no additional smoothing mechanism is available from the second-order It\^o correction term.
\end{remark}

\subsection{Evolution of the covariance tensor}
We have seen that the expectation of the variables in the LA SALT equation form a closed system. Could we say the same about the covariance? For general semi-direct product LA-SALT systems \eqref{LASALT-adv-form}, the answer is no. However, the covariance for the advected quantities {\em does} always form a closed system as we will show below.
%(\textcolor{red}{A: I would make a bit more clear that you are referring to general semidirect product systems are not particularly to Boussinesq}).

%\todo[inline]{DH: So, how about trying to get a closed system for the variance of the sum $\|\omega'\|_{L^2}^2+\|\nabla \theta'\|_{L^2}^2$? }

%\textcolor{red}{ST: I tried this already but the $\nabla \theta$ doesn't form a closed system because $\nabla$ and Lie derivative doesn't commute. You get an extra commutator term $[\nabla, \mathcal L]$ if you try to do this, which doesn't close}

%\todo[inline]{DH: Darn! I would have thought differently, since differential and Lie derivative commute. }

%\textcolor{red}{ST: That's actually a really good point, I'll look into it again}

%\todo[inline]{DH: Great! If that works, the diamond term may show the way to do it for an arbitrary advected quantity. Or maybe, it'll only work for scalars appearing linearly in the Lagrangian and Hamiltonian. Good luck!}

\begin{proposition} \label{covariance-eq}
Let $a_t$ be any tensor field that satisfies the linear stochastic advection equation
\begin{align}\label{a-trans-eq}
    \diff a + \mathcal L_{\E{\frac{\delta H}{\delta \mu}}} a \diff t + \displaystyle\sum_k \mathcal L_{\xi_k} a \circ \diff W_t^k = 0,
\end{align}
and let $A^{(2)} := \mathbb E\left[(a - \mathbb E[a])^2\right]$ be the covariance tensor for the tensor field $a$, where $\bullet^2$ here means taking the tensor product with itself.
Then $A^{(2)}$ satisfies the following PDE:
\begin{align}\label{covar-eq}
\partial_t A^{(2)} + \mathcal L_{\E{\frac{\delta H}{\delta \mu}}}A^{(2)} = \sum_k \left(\frac12 \mathcal L_{\xi_k}^2 A^{(2)} + \left(\mathcal L_{\xi_k} \E{a}\right)^2 \right).
\end{align}
This is closed since $\E{a}$ and $\E{\frac{\delta H}{\delta \mu}}$ are determined by the closed system \eqref{LASALT-EPX}. 
\end{proposition}
\begin{proof}[Proof of Proposition \ref{covariance-eq}]
Let $a' := a - \E{a}$ be the fluctuation about the mean, which can be shown using \eqref{a-trans-eq}-\eqref{LASALT-EPX} to satisfy
\begin{align} \label{a-fluctuation-eq}
    \diff a' + \mathcal L_{\E{\frac{\delta H}{\delta \mu}}} a' \diff t + \displaystyle\sum_k \mathcal L_{\xi_k} a \circ \diff W_t^k = -\frac12 \displaystyle\sum_k \LL_{\xi_k}^2 \E{a} \diff t.
\end{align}
Then by It\^o's product rule, we have
\begin{align*}
    \diff \, (a')^2 &= (\circ \diff a') \otimes a' + a' \otimes (\circ \diff a') \\
    &= -\mathcal L_{\E{\frac{\delta H}{\delta \mu}}} a' \otimes a' \diff t - \displaystyle\sum_k \mathcal L_{\xi_k} a \otimes a' \circ \diff W_t^k - \frac12 \displaystyle\sum_k \mathcal L_{\xi_k}^2 \E{a} \otimes a' \diff t \\
    &\quad\, - a' \otimes \mathcal L_{\E{\frac{\delta H}{\delta \mu}}} a'  \diff t - \displaystyle\sum_k  a' \otimes \mathcal L_{\xi_k} a \circ \diff W_t^k - \frac12 \displaystyle\sum_k a' \otimes \mathcal L_{\xi_k}^2 \E{a} \diff t,
\end{align*}
and using the Leibniz property of the Lie derivative, i.e., $\LL(S \otimes T) = \LL S \otimes T + S \otimes \LL T$, for any tensors $S$ and $T$, we have
\begin{align}
    &\diff \,(a')^2 + \mathcal L_{\E{\frac{\delta H}{\delta \mu}}} (a')^2 \diff t + \frac12 \displaystyle\sum_k \left(\mathcal L_{\xi_k}^2 \E{a} \otimes a' + a' \otimes \mathcal L_{\xi_k}^2 \E{a}\right) \diff t \nonumber\\
    &= -\displaystyle\sum_k \left(\mathcal L_{\xi_k} a \otimes a' + a' \otimes \mathcal L_{\xi_k} a\right) \circ \diff W_t^k \nonumber\\
    &= \frac12 \displaystyle\sum_k \left(\mathcal L_{\xi_k}^2 a \otimes a' + 2\left(\mathcal L_{\xi_k} a\right)^2 + a' \otimes \mathcal L_{\xi_k}^2 a\right) \diff t - \displaystyle\sum_k \left(\mathcal L_{\xi_k} a \otimes a' + a' \otimes \mathcal L_{\xi_k} a\right) \diff W_t^k \nonumber\\
    &= \frac12 \displaystyle\sum_k \underbrace{\left(\mathcal L_{\xi_k}^2 a' \otimes a' + 2\left(\mathcal L_{\xi_k} a'\right)^2 + a' \otimes \mathcal L_{\xi_k}^2 a'\right)}_{\hbox{$= \LL_{\xi_k}(\mathcal L_{\xi_k} (a')^2)$}} \diff t  + \sum_k\left(\mathcal L_{\xi_k} \E{a}\right)^2 \diff t -\displaystyle\sum_k \left(\mathcal L_{\xi_k} a \otimes a' + a' \otimes \mathcal L_{\xi_k} a\right) \diff W_t^k \nonumber\\
    &\quad  +\frac12 \displaystyle\sum_k \left(\mathcal L_{\xi_k}^2 \E{a} \otimes a' + 2\left(\mathcal L_{\xi_k} \E{a} \otimes \mathcal L_{\xi_k} a' + \mathcal L_{\xi_k} a'  \otimes \mathcal L_{\xi_k} \E{a} \right) + a' \otimes \mathcal L_{\xi_k}^2 \E{a}\right) \diff t,
\label{a'-squared-eq}
\end{align}
where in the second equality we converted from Stratonovich to It\^o integral (see appendix \ref{strat-ito-correction}) and in the last equality, we expanded the Stratonovich-to-It\^o correction term using $a = a' + \E{a}$ and the linearity of Lie derivatives. Taking expectations on both sides of \eqref{a'-squared-eq} and noting that (1) the expectation of the It\^o integral vanishes by the martingale property, and (2) $\E{a'} = 0$ by definition, we obtain
\begin{align*}
    \partial_{t} A^{(2)} + \mathcal L_{\E{\frac{\delta H}{\delta \mu}}}A^{(2)} = \displaystyle\sum_k \left(\frac12 \mathcal L_{\xi_k}^2 A^{(2)} + \left(\mathcal L_{\xi_k} \E{a}\right)^2 \right),
\end{align*}
as expected, where $A^{(2)} = \E{(a')^2}$.
\end{proof}

The covariance for the $\mu$ variable in \eqref{LASALT-adv-form} is unlikely to form a closed equation in general due to presence of the coupling term $\E{\frac{\delta H}{\delta a}} \diamond a$, however in the special example of the 2D Boussinesq equation, this is indeed possible as we will illustrate in the next example.

\begin{example}[Covariance of 2D LA-SALT Boussinesq]\label{2DEB-example}
Let us consider the special case of 2D LA-SALT Boussinesq system \eqref{LA:SALT:Ito:Bou}. Letting $u' := u - \E{u}$ and $\theta' := \theta - \E{\theta}$, we have the following equations for the fluctuations
\begin{equation}\label{theta-prime-eq}
\left\{
\begin{array}{rl}
\diff u' + \mathcal L_{\E{u}} u' \diff t + \displaystyle\sum_{k}\mathcal L_{\xi_k} u \circ \diff W_t^k &= -\displaystyle\sum_{k}\frac12 \mathcal L_{\xi_k}^2 \E{u} \diff t - g y {\sf d} \theta' \diff t, \\
    \diff \theta' + \mathcal L_{\E{u}} \theta' \diff t + \displaystyle\sum_{k}\mathcal L_{\xi_k} \theta \circ \diff W_t^k &= -\displaystyle\sum_{k}\frac12 \mathcal L_{\xi_k}^2 \E{\theta} \diff t. 
\end{array}
\right.
\end{equation}

Then by similar arguments as in the proof of Proposition \ref{covariance-eq}, we can show that the covariance fields satisfy the following PDEs
\begin{equation}\label{theta-var-eq}
\left\{
\begin{array}{rl}
    \partial_{t} U^{(2)}+ \mathcal L_{\E{u}}U^{(2)} &= \displaystyle\sum_k \left(\frac12 \mathcal L_{\xi_k}^2 U^{(2)} + \left(\mathcal L_{\xi_k} \E{u}\right)^2\right) - gy \E{u' \otimes {\sf d} \theta' + {\sf d} \theta' \otimes u'},  \\
    \partial_{t}\Theta^{(2)}+ \mathcal L_{\E{u}}\Theta^{(2)} &= \displaystyle\sum_k \left(\frac12 \mathcal L_{\xi_k}^2\Theta^{(2)} + \left(\mathcal L_{\xi_k} \E{\theta}\right)^2\right),
\end{array}
\right.
\end{equation}
where $U^{(2)} := \E{(u')^2}$ and $\Theta^{(2)} := \E{(\theta')^2}$. Clearly, this system is not closed due to the presence of the term $\E{u' \otimes {\sf d} \theta' + {\sf d} \theta' \otimes u'}$ in the $U^{(2)}$ equation. However, applying the exterior derivative $\sf d$ on both sides of the $\theta$-equation and its corresponding fluctuation \eqref{theta-prime-eq}, and noting that the exterior derivative and the Lie derivative commute as a consequence of Cartan's formula, we obtain the following system for $\sf{d}\theta$ and $\sf{d}\theta'$:
\begin{equation}
    \left\{
    \begin{array}{rl}
&\diff \,({\sf d} \theta) + \mathcal L_{\E{u}} {\sf d} \theta \diff t + \displaystyle\sum_k\mathcal L_{\xi_k} {\sf d} \theta \circ \diff W_t^k = 0, \\
&\partial_{t}\E{{\sf d} \theta} + \mathcal L_{\E{u}} {\sf d} \E{ \theta} = \frac12 \displaystyle\sum_k \mathcal L_{\xi_k}^2 \E{{\sf d}\theta},\\
&\diff\, ({\sf d} \theta') + \mathcal L_{\E{u}} {\sf d} \theta' \diff t + \displaystyle\sum_k\mathcal L_{\xi_k} {\sf d} \theta \circ \diff W_t^k = -\frac12\displaystyle\sum_k \mathcal L_{\xi_k}^2 \E{{\sf d}\theta} \diff t.
\end{array}
\right.
\end{equation}

By Proposition \ref{covariance-eq}, the covariance for $\sf{d}\theta$ evolves as
\begin{align}
    \partial_{t}({\sf{d}}\Theta^{(2)})+ \mathcal L_{\E{u}}{\sf{d}}\Theta^{(2)} = \displaystyle\sum_k \left(\frac12 \mathcal L_{\xi_k}^2 {\sf{d}}\Theta^{(2)} + \left(\mathcal L_{\xi_k} \E{{\sf{d}}\theta}\right)^2\right),\label{d-theta-var-eq}
\end{align}
where ${\sf{d}}\Theta^{(2)} := \E{(\sf{d}\theta')^2}$. We show that obtaining an equation for $\E{u' \otimes {\sf d} \theta' + {\sf d} \theta' \otimes u'}$ closes the system \eqref{theta-var-eq}.

By the stochastic product rule, we have
\begin{align}
    &\diff \,(u' \otimes {\sf d}\theta') = u' \otimes (\circ \diff \,({\sf d} \theta')) + (\circ \diff u') \otimes {\sf d} \theta'\nonumber \\
    &=-\mathcal L_{\E{u}}(u' \otimes {\sf d}\theta') \diff t - \left(u'\otimes \mathcal L_{\xi_k}({\sf d}\theta) + \mathcal L_{\xi_k} u \otimes {\sf d}\theta'\right)\circ \diff W_t^k \nonumber\\
    &\qquad - \frac12 \left(u' \otimes \mathcal L_{\xi_k}^2 \E{{\sf d}\theta} + \mathcal L_{\xi_k}^2 \E{u} \otimes {\sf d}\theta'\right) - gy ({\sf{d}}\theta')^2\diff t \nonumber\\
    &=-\mathcal L_{\E{u}}(u' \otimes {\sf d}\theta') \diff t - \left(u'\otimes \mathcal L_{\xi_k}({\sf d}\theta) + \mathcal L_{\xi_k} u \otimes {\sf d}\theta'\right) \diff W_t^k \nonumber\\
    &\qquad - \frac12 \left(u' \otimes \mathcal L_{\xi_k}^2 \E{{\sf d}\theta} + \mathcal L_{\xi_k}^2 \E{u} \otimes {\sf d}\theta'\right) - gy ({\sf{d}}\theta')^2\diff t \nonumber \\
    &\qquad \qquad + \frac12 \left(u' \otimes \mathcal L_{\xi_k}^2 ({\sf d} \theta) + 2(\mathcal L_{\xi_k} u) \otimes \left(\mathcal L_{\xi_k} ({\sf d} \theta)\right) + \mathcal L_{\xi_k}^2 u \otimes {\sf d} \theta'\right) \diff t, \label{omega-theta-prime}
\end{align}
By the Leibniz property of Lie derivatives, we have
\begin{align*}
    &\mathcal L_{\xi_k} \mathcal L_{\xi_k} \left(u' \otimes {\sf{d}}\theta'\right) = \mathcal L_{\xi_k} \left(u' \otimes \mathcal L_{\xi_k}({\sf{d}}\theta') +\mathcal L_{\xi_k}(\omega') \otimes {\sf{d}}\theta'\right) \\
    &= u' \otimes \mathcal L_{\xi_k}^2 ({\sf d} \theta') + 2(\mathcal L_{\xi_k} u') \otimes \left(\mathcal L_{\xi_k} ({\sf d} \theta')\right) + \mathcal L_{\xi_k}^2 \omega' \otimes {\sf d} \theta'.
\end{align*}

Now using the above expression and taking expectations on both sides of \eqref{omega-theta-prime} give us the PDE:
\begin{align*}
    \begin{split}
   \partial_{t} \E{u' \otimes \sf d\theta'} + \mathcal L_{\E{u}}\E{u' \otimes {\sf d}\theta'}
    &= \frac12 \mathcal L_{\xi_k}^2 \E{u' \otimes  {\sf{d}}\theta'} \\
    &\qquad + \left(\mathcal L_{\xi_k} \E{u}\right) \otimes \left(\mathcal L_{\xi_k} \E{{\sf d} \theta}\right) - gy\E{({\sf{d}}\theta')^2}.
    \end{split}
\end{align*}
Similarly, we get an equation for $\E{{\sf{d}}\theta' \otimes u'}$ and combining them gives us an equation for $\E{u' \otimes {\sf d} \theta' + {\sf d} \theta' \otimes u'}$, which reads
\begin{align}\label{cross-variance-eq}
    \begin{split}
    &\partial_{t} \E{u' \otimes \sf d\theta' + {\sf d} \theta' \otimes u'} + \mathcal L_{\E{u}}\E{u' \otimes {\sf d}\theta' + {\sf d} \theta' \otimes u'}
    = \frac12 \mathcal L_{\xi_k}^2 \E{u' \otimes  {\sf{d}}\theta' + {\sf d} \theta' \otimes u'} \\
    &\qquad + \left(\mathcal L_{\xi_k} \E{u}\right) \otimes \left(\mathcal L_{\xi_k} \E{{\sf d} \theta}\right) + \left(\mathcal L_{\xi_k} \E{{\sf d} \theta}\right) \otimes \left(\mathcal L_{\xi_k} \E{u}\right) - 2gy\E{({\sf{d}}\theta')^2}.
    \end{split}
\end{align}

Since the last term  $\E{({\sf{d}}\theta')^2}$ is just the covariance tensor ${\sf{d}}\Theta^{(2)}$, which we can solve for, we conclude that equations \eqref{theta-var-eq},\eqref{d-theta-var-eq} and \eqref{cross-variance-eq} form together a closed system for the covariance of the fields $(u, \theta)$ in the 2D LA-SALT Boussinesq system.
\end{example}

\begin{remark}
By having a closed system of PDEs for the evolution of the covariance, we may deduce for instance its growth behaviour through the application of standard PDE methods. For instance if we consider the equation for the evolution of $\Theta^{(2)}$ \eqref{LASALT-adv-form}, where we assume incompressibility $\div{(\E{\u})} = 0,$ and choose $\xi^{(1)} = \hat{\mb{x}},$ $\xi^{(2)}=\hat{\mb{y}}$, then we can check directly that its $L^2$-norm satisfies
\begin{align*}
    \|\Theta^{(2)}_t\|_{L^2}^2 + \frac12 \int^t_0 \|\nabla \Theta^{(2)}_s\|_{L^2}^2 \diff s = \int^t_0 \|\nabla \E{\theta_s}\|_{L^2}^2 \diff s,
\end{align*}
where we have taken into account that $\Theta^{(2)}(0) = 0$. Since by the parabolicity of the expectation equation \eqref{LASALT-EPX}, we have the estimate
\begin{align*}
    \int^T_0 \|\nabla \E{\theta_t}\|_{L^2}^2 \diff t \leq C \|\theta_0\|_{L^2}^2 \, e^T,
\end{align*}
where $C > 0$ is some constant (see Section \ref{expectation-proof} below for more details), we can deduce that the space-averaged covariance $\|\Theta^{(2)}\|_{L^2}^2$ evolves at most exponentially fast.
\end{remark}

\begin{remark}[Extension to $p$-th central moments]\label{p-th-moment-remark}
One may also ask if closed equations for the higher moments of the advected tensor field $a$ can be derived, thus providing a generalisation of Proposition \ref{covariance-eq}, which may help us to understand for instance the non-Gaussianity of the system. In the case where $a$ is a scalar field, the $p$-th central moment $A^{(p)} := \E{(a-\E{a})^p}$ indeed satisfies a closed, iterated system:
\begin{align} \label{pth-moment}
    \begin{split}
   \partial_{t} A^{(p)} + \mathcal L_{\E{\frac{\partial H}{\partial \mu}}}A^{(p)} &= \displaystyle\sum_k \left(\frac12 \mathcal L_{\xi_k}^2 A^{(p)} + p\left(\mathcal L_{\xi_k} A^{(p-1)}\right)\left(\mathcal L_{\xi_k} \E{a}\right) + \frac{p(p-1)}{2}A^{(p-2)} \left(\mathcal L_{\xi_k} \E{a}\right)^2\right),
    \end{split}
\end{align}
which recovers \eqref{covar-eq} in the case $p=2$ (see Appendix \ref{p-th-moment-scalar} for the proof). However, when $a$ is a general tensor field, we have not been able to obtain a closed system for its $p$-th central moment due to the non-commutativity of the tensor product (which is commutative only in the scalar field case).
\end{remark}

\section{Preliminaries and notation for the analysis in Section \ref{sec:well-posedness}}\label{sec:prelim-anal}

\subsection{Function spaces, inequalities and embeddings}
We define the $L^2$-inner product as $\left<f,g\right>_{L^2}:= \int_{\mathbb{T}^{2}} f(x)\cdot g(x) \ {\sf d} V$, where ${\sf d} V$ is the Lebesgue measure on $\mathbb{T}^{2}$ and denote by $\|\cdot\|_{L^2}$ its corresponding norm. For any $k\in \mathbb N$, we denote the Sobolev space by $H^{k}$, equipped with the norm
\[
||f||^{2}_{H^{k}}:=\sum_{j \leq k} \|D^jf\|_{L^2}^2,
\]
where $D$ represents weak derivative. We define the space $\dot{H}^{k}$ to be the subspace of $H^k$ consisting of functions that integrate to zero, that is, $\int_{\mathbb{T}^{2}} f(x) \ {\sf d} V=0$. For any $p\in \mathbb{Z}$, we denote by $L^{p}(\mathbb{T}^{2}; \mathbb{R}^{2})$ the class of all measurable $p$ - integrable functions defined on the two-dimensional torus, with values in $\mathbb{R}^{2}$ . This space is endowed with its canonical norm $\|f\|_{L^p} := \bigg(\displaystyle\int_{\mathbb{T}^2} \abs{f}^{p} \ \sf d V\bigg)^{1/p}$. Conventionally, for $p = \infty$ we denote by $L^{\infty}$ the space of essentially bounded measurable functions. Next, if $X$ is a general Banach space we let $C([0, \infty); X)$ be the space of continuous functions from $[0, \infty)$ to $X$ equipped with 
the uniform convergence norm over compact subintervals of $[0, \infty)$ and
$L^p([0, T]; X)$ the space of measurable functions from $[0, \infty)$ to $X$ such that the norm 
\[
\|f\|_{L^p(0, T; X)}=\bigg(\displaystyle\int_0^t \|f(t)\|_X^p \diff t\bigg)^{1/p}
\] 
is finite. \\ 
\textbf{Biot-Savart operator.} For any $\f:\mathbb{T}^{2}\to\mathbb{R}^{2}$, we define the curl operator
\[ \text{curl} \ \f=\nabla^{\perp}\cdot \f=\partial_{2}f^{2}-\partial_{1}f^{1}, \]
where $\partial_{i}$ denotes the derivative $\frac{\partial}{\partial {x^{i}}}$. The inverse of the curl operator (known as the Biot-Savart operator) is defined by
\[ \mathcal{K}f:=\nabla^{\perp}(-\Delta)^{-1} f \]
acting on mean-free functions $f:\mathbb{T}^{2}\to \mathbb{R}$. In fluid dynamics, the Biot-Savart operator allows us to reconstruct the mean-free component of the velocity vector field $\u$ from the vorticity function $\omega$, satisfying $\text{curl} \ \u=\omega$. Moreover, we have the inequality
\begin{equation}\label{BiotSavart:ine}
\left\Vert \u\right\Vert _{H^{k+1}}\leq C_{k}\left\Vert \omega\right\Vert
_{H^{k}},
\end{equation}
for all $k\geq 0,$ where $C_{k}>0$ represents a positive constant depending only on $k$, cf. \cite{Scho}. \\
\textbf{Inequalities and embeddings.} 
In our well-posedness analysis below, we will be using different forms of Sobolev embeddings, namely, the Gagliardo-Nirenberg interpolation inequalities. For the sake of clarity, we list below the ones we will make use of most often. For every smooth function $f:\mathbb{T}^{2}\to \mathbb{R}$ with zero-mean, it holds
\begin{eqnarray}
\left\Vert f \right\Vert_{L^{3}} &\lesssim& \left\Vert f \right\Vert ^{2/3}_{L^{2}}  \left\Vert \nabla f \right\Vert ^{1/3}_{L^{2}}, \label{Sob:ine1} \\
\left\Vert f \right\Vert_{L^4} &\lesssim& \norm{\nabla f }_{L^{2}}. \label{Sob:ine2} 
\end{eqnarray} 
 
We will also need a particular case of Young's inequality, which states that for any $a,b\in\RR^{+}$ and $1\leq p,q \leq \infty$ such that $1/p + 1/q=1$, it holds
\[ ab\leq \frac{a^p}{p} + \frac{b^q}{q}.\] 
A particular case of this reads
\begin{equation}\label{PP:ine}
ab\leq  \frac{a^2}{2\epsilon} +\frac{\epsilon b^2}{2}, \quad \text{for}\quad \epsilon >0,
\end{equation}
which is typically referred to as Peter-Paul's inequality.

\subsection{Some results from stochastic analysis}
We recall some results from the theory of stochastic processes that will be employed in our proofs later. We refer the reader to the standard references \cite{PraZab,Fla96} for a more thorough review. We begin by fixing a stochastic basis
 $\mathcal{S}=(\Xi,\mathcal{F},\lbrace\mathcal{F}_{t}\rbrace_{t\geq 0}, \mathbb{P},\lbrace W^{i}\rbrace_{i=1}^N),$ that is, a filtered probability space $(\Xi,\mathcal F,\mathbb P)$ together with a family $\lbrace W^{i} \rbrace_{i = 1}^N$ of i.i.d. Brownian motions that is adapted to the filtration $\lbrace\mathcal{F}_{t}\rbrace_{t\geq0}$.  \\ \\
In our existence proof for the linear stochastic equations, we will use the following version of It\^o's lemma:
\begin{lemma}[It\^o's first formula, \cite{kunita1997stochastic}]\label{ito-lemma}
Let $\phi_t$ be the flow of the following forward Stratonovich SDE
\begin{align*}
    \diff X_t = \mu(t,X_t) \diff t + \displaystyle\sum_k \sigma_k(t,X_t) \circ \diff W_t^k.
\end{align*}
Then for any $C^2$-smooth $k$-form $K$, we have the following
\begin{align}\label{ito-first-formula}
    \begin{split}
    (\phi_t)_*K(x) - K(x) &= - \int^t_0 \LL_\mu\left((\phi_s)_*K(x)\right) \diff s - \displaystyle\sum_k \int^t_0 \LL_{\sigma_k}\left((\phi_s)_*K(x)\right) \diff W_s^k \\
    &\quad + \frac12 \displaystyle\sum_k \int^t_0 \LL_{\sigma_k}^2 ((\phi_s)_*K(x))\diff s,
    \end{split}
\end{align}
where $(\phi_t)_*$ denotes push-forward (right action by the inverse of $\phi_t$).
\end{lemma}

We also use the following lemma, which is an easy corollary of the Kunita-It\^o-Wentzell formula stated in \cite{de2019implications}, Theorem 3.1.
\begin{lemma} \label{ito-int}
Let $\phi_t$ be the flow as in Lemma \ref{ito-lemma}. Then for any $C^2$ semimartingale $K$, taking values in the $k$ forms, we have
\begin{align}
    \begin{split}
    &(\phi_t)_*\left(\int^t_0 K(s,x) \diff s\right) = \int^t_0 (\phi_s)_*K(s,x) \diff s - \int^t_0 \LL_\mu\left(\int^s_0(\phi_s)_*K(r,x)\diff r\right) \diff s \\
    &- \displaystyle\sum_k \int^t_0 \LL_{\sigma_k}\left(\int^s_0(\phi_s)_*K(r,x)\diff r\right) \diff W_s^k + \frac12 \displaystyle\sum_k \int^t_0 \LL_{\sigma_k}^2 \left(\int^s_0(\phi_s)_*K(r,x)\diff r\right)\diff s.
    \end{split}
\end{align}
\end{lemma}

When obtaining estimates for It\^o integrals, the Burkholder-Davis-Gundy inequality will be needed. In the present context, it reads
\begin{equation}\label{BGD:ineq}
 \mathbb{E}\left[ \displaystyle \sup_{s \in [0,T]}\left| \int_{0}^{t} X_{s} \diff W_ {s}\right|^{p} \right] \leq C_p \mathbb{E} \left[ \int_{0}^{T} |X_s|^{2} \ \diff t \right]^{p/2},
\end{equation}
for any $p\geq 1,$ where $C_{p}$ is an absolute constant depending on $p$. Here, $X_t$ is any square integrable semimartingale adapted to the filtration $\{\mathcal F_t\}_{t\geq0}$.
\subsection{Assumptions on the noise vector fields $\{\xxi_{k}\}_{k=1}^N$}
In the well-posedness analysis, we assume that
the vector fields $\xxi_{k}:\TT^2 \rightarrow \RR^2, \,k=1,\ldots,N$ are of class $L^\infty([0,T],C^{4+\alpha}(\TT^2,\RR^2))$ for some $0<\alpha<1$ and satisfy the uniform ellipticity condition
\begin{equation}\label{elliptic:Lie}
\displaystyle\sum_{k=1}^N \displaystyle\sum_{i,j=1}^{2} \xi^{i,j}_{k}(t,x)\xi^{i,j}_{k}(t,x)\eta_{i}\eta_{j}\geq \lambda  |\eta|^{2},
 \end{equation}
for some $\lambda>0$ and every $\eta \in\RR^2$. This generalises a fundamental property of the Laplace operator $\Delta$, which can be recovered by choosing ${\mb \xi}_k = e_k$, $k=1,2,$ where $\{e_1,e_2\}$ is the canonical basis for $\RR^2$.

From an analysis perspective, we opt to work with Lie derivatives instead of general first order differential operators since the curl commutes with the Lie derivative (a consequence of Cartan's formula), allowing us to obtain the vorticity formulation of the stochastic Boussinesq system, making the analysis simpler. 
\subsection{Statements of the main analytical results}
Let us state here the notion of solution we will employ and the main theorems that we are going to prove in the following sections. For this, we need to understand the strategy we are going to follow in order to solve the 2D LA SALT Boussinesq equations $\eqref{LA:SALT:Ito:Bou}.$ Indeed, to construct a solution, we carry out the following steps:\\
\begin{itemize}
    \item{We start by solving for the variables $\mathbb{E}[\u]$ and $\mathbb{E}[\theta]$ in the equations for the expectation \eqref{LA:SALT:Exp:Bou}. Since $\mathbb{E}[\u]$ is incompressible, it is sufficient to project the momentum equation onto its vorticity formulation by applying the curl operator, solve for the new system, and then recover $\mathbb{E}[\u]$ by means of the Biot-Savart law. The pressure terms $\mathbb E[p-|\u|^2/2]$ can be recovered by solving a Poisson problem for the pressure terms. This argument is detailed in  Section \ref{sec:well-posedness}}
    \item{We plug the already solved deterministic variables $\mathbb{E}[\u],$ $\mathbb{E}[\theta],$ and $\mathbb E[p-|\u|^2/2]$ into the main equations \eqref{LA:SALT:Ito:Bou}. We note that these become stochastic linear equations with a forcing. Now we need to solve for $\u.$ To this aim, we provide the required estimates on the linear equation. Once $\u$ is solved, we can recover the expected pressure as $\mathbb{E} [p] = \mathbb{E} [p - |\u|^2/2] +  \mathbb{E} [|\u|^2/2].$ We will denote $f(t,x) = -{\sf{d}}\mathbb{E} [p - |\u|^2/2-g\theta y] + g\E{\theta}\mb{\hat{y}}$ since this appears as a forcing after solving the pressure term in the expected equations. }
    %\item{To conclude our argument, we have to solve for $\mathbb{P} \u$ and $\mathbb{Q} \u.$ For this, we observe that since $\mathbb{P} \u$ is the divergence-free part of $\u,$ we can take curls in equation \eqref{LA:SALT:Ito:Bou}, solve for the vorticity $\omega = curl(\u),$ and recover $\mathbb{P} \u$ via Biot-Savart.}
    %\item{Finally, we also need to solve for the gradient part of $\u,$ $\mathbb{Q} \u$. For this, we observe that $\mathbb{Q} \u = \mathbb{Q}(\u - \mathbb{E} [\u]) + \mathbb{Q} \mathbb{E} [\u] = \mathbb{Q}(\u - \mathbb{E} [\u]),$ where we have used linearity and the divergence-free mean condition expressed as $\mathbb{P} \mathbb{E} \u = \mathbb{P} \u $ or $\mathbb{Q} \mathbb{E} \u = 0. $ This simplifies the computations because it annihilates the deterministic pressure.}

\end{itemize}

\begin{definition}[Strong solution of the 2D LA-SALT-Boussinesq equations]
We say that a process $(u, \theta) \in L^2\left(\Omega; C(\mathbb R^+, H^2(\mathbb T^2, \mathbb R^2) \times H^3(\mathbb T^2, \mathbb R))\right) $ is a strong global solution to the 2D LA-SALT-Boussinesq equation \eqref{LA:SALT:Ito:Bou} if $(u_t, \theta_t)$ is adapted to the filtration $\{\mathcal F_t\}_{t \geq 0}$ and satisfies
\begin{align}\label{LA-Boussinesq-vorticity}
    \begin{cases}
    &u(t,x,y) - u_0(x,y) + \displaystyle\int^t_0 \mathcal L_{\mathbb E[u]} u(s,x,y) \diff s + \displaystyle\sum_{k=1}^N \int^t_0 \mathcal L_{\xi_k} u(s,x,y) \,\diff W_s^k + \int_0^t f(s,x,y) \diff s  \\
    &\qquad =  \displaystyle \frac12 \sum_{k=1}^N \int^t_0 \mathcal L_{\xi_k}^2 u(s,x,y) \diff s - gy \int_0^t {\sf d}\theta(s,x,y) \diff s, \\
    &\theta(t,x,y) - \theta_0(x,y) + 
\, \displaystyle\int^t_0 \mathcal L_{\mathbb E[u]} \theta(s,x,y) \diff s + \displaystyle\sum_{k=1}^N \int^t_0 \mathcal L_{\xi_k} \theta(s,x,y) \,\diff W_s^k \\
& \quad \displaystyle =  \frac12\sum_{k=1}^N \int^t_0 \mathcal L_{\xi_k}^2 \theta(s,x,y) \diff s,
    \end{cases}
\end{align}
almost surely for all $t>0$, where, as specified in the preliminaries, the notation $\mathcal{L}$ is used to represent both the Lie derivative applied to one-forms $u$ and to scalars $\theta$.
\end{definition}

\begin{theorem}\label{main:th:1}
Let $(\u_0,\theta_0) \in H^2(\mathbb T^2, \mathbb R^2) \times H^3(\mathbb T^2, \mathbb R)$, then there exists a unique global strong solution of the 2D LA-SALT Boussinesq equation \eqref{LA:SALT:Ito:Bou}.
\end{theorem}

%and $\xi_i \in L^\infty(\mathbb R, C^{4+\alpha}(\mathbb T^2, \mathbb R^2))$ for some $0 < \alpha < 1$ and $i=1,\ldots,N$,
The proof of Theorem \ref{main:th:1} strongly depends on the following fact, which we state as a separate result:

\begin{theorem} \label{main:th:2}
Let $(\mathbb{E}[\u_{0}],\mathbb{E}[\theta_{0}])\in H^{5}(\mathbb T^2, \mathbb R^2)\times H^{3}(\mathbb T^2, \mathbb R)$ be an initial data. Then equations \eqref{LA:Vor:Exp:Bou} have a unique global strong solution $(\mathbb{E}[\u],\mathbb{E}[\theta])\in C\left([0,\infty), H^{5}(\mathbb T^2, \mathbb R^2)\times H^{3}(\mathbb T^2, \mathbb R)\right)$.
\end{theorem}
\begin{remark}
 We note that even though we start the stochastic equation with $\u_0 \in H^2(\mathbb T^2, \mathbb R^2)$, we require an additional assumption that $\E{\u_0} \in H^5(\mathbb T^2, \mathbb R^2)$ in order to prove Theorem \ref{main:th:1}. This is possible since the expectation of a random field may be smoother than the random field (for example take $u_0(x) = x + W_x$, where $W_x$ is spacial Brownian motion. Then $u_0$ is not even differentiable, yet $\E{u_0(x)} = x$ is smooth.) For deterministic initial conditions, this means that we start with $u_0 \in H^5(\mathbb T^2, \mathbb R^2)$ but it loses regularity to $H^2(\mathbb T^2, \mathbb R^2)$ as soon as $t>0$ and will remain there. We may also construct weak solutions instead of classical solutions by employing the techniques in \cite{de2019well}, thus avoiding this assumption altogether.
\end{remark}
\begin{remark}
The regularity of the initial datum in Theorem \ref{main:th:2} is not sharp. It is well-known that even for $L^{2}$ regular initial data we could still provide the instantaneous regularisation of the solutions. We have chosen precisely $H^{5}\times H^{3}$ since our final goal is to prove Theorem \ref{main:th:1} which requires higher regularity. This is due to the need of having a sufficiently smooth coefficient for the characteristic equations in order get a smooth flow (cf. Subsection \ref{Thm-1-proof}). It is easy to check that equations \eqref{LA:SALT:Ito:Bou} lose their parabolic character and become pure transport equations. 
\end{remark}

\section{Global well-posedness of the 2D LA-SALT Boussinesq system} \label{sec:well-posedness}

\subsection{Well-posedness of the expectation equations}\label{expectation-proof}
In this section we will provide the proof of Theorem \eqref{main:th:2}. The strategy of the proof will be divided into three parts: first, we will use energy methods to provide a priori estimates of the solution. Next, we will show a bound for the evolution of the average of $\U$. Finally, we will prove the uniqueness of solutions. \\\\
\textbf{Step 1: Energy methods and a priori estimates.}
First of all, to simplify the exposition we will employ the following notations: 
\[ \mathbb{E}[\u]:= \U, \quad \mathbb{E}[\theta]:=\Theta, \quad \mathbb{E}[p]:=P, \]
which we decompose as
\[ \U= \overbar{\U}+\V, \quad \Theta=\overbar{\Theta}+W, \quad \text{such that} \ \quad \overbar{\U}=\int_{\TT^2}\U \ {\sf d} V \quad \text{and} \quad  \overbar{\Theta}=\int_{\TT^2}\Theta \ {\sf d} V, \]
where $\V,W$ are mean-free . With this notation at hand, equation \eqref{LA:SALT:Exp:Bou} is given by
\begin{equation}\label{new:expectation:bouss}
\left\{
\begin{array}{rl}
 \partial_{t} \U+ \mathcal{L}_{\U}\U = g\Theta \hat{y}& -\nabla \left(P - \mathbb{E} \left[ \frac{|\u|^{2}}{2} \right] \right) + \displaystyle \frac{1}{2} \sum_{k=1}^{N} \mathcal{L}^{2}_{\xi_k}\U,  \\ \partial_{t}\Theta+\mathcal{L}_{\U}\Theta&=\displaystyle \frac{1}{2} \sum_{k=1}^{N} \mathcal{L}^{2}_{\xi_k}\Theta,
\end{array}
\right.
\end{equation} 
and by applying the curl to the momentum equation above and taking into account that the curl and the Lie derivative commute, we obtain the following coupled PDE system for the vorticity and the potential temperature
\begin{equation}\label{new:expectation:vor:bouss}
\left\{
\begin{array}{rl}
 \partial_{t} \Omega + \mathcal{L}_{\U}\Omega&= g\p_{x}\Theta + \displaystyle \frac{1}{2} \sum_{k=1}^{N} \mathcal{L}^{2}_{\xi_k}\Omega, \\ \partial_{t}\Theta+\mathcal{L}_{\U}\Theta&=\displaystyle \frac{1}{2} \sum_{k=1}^{N} \mathcal{L}^{2}_{\xi_k}\Theta,
 \end{array}
\right.
\end{equation} 
where $\Omega:=\nabla^{\perp}\cdot\U$. Notice also that by integrating \eqref{new:expectation:bouss} in space, the mean of $\U$ evolves as
\begin{equation}\label{mean:equation:evolution}
 \frac{\diff \overbar{\U}}{\diff t} = g\int_{\TT^{2}}\Theta \hat{y} \ {\sf d} V + \displaystyle \frac{1}{2} \sum_{k=1}^{N} \int_{\TT^{2}}  \mathcal{L}^{2}_{\xi_k}\V \ {\sf d} V +\displaystyle \frac{1}{2} \sum_{i=1}^{N} \int_{\TT^{2}} \mathcal{L}^{2}_{\xi_k}\overbar{\U} \ {\sf d} V.
\end{equation} 
Let assume that $(\Omega(t),\Theta(t))$ are smooth and satisfy \eqref{new:expectation:vor:bouss}. Then for $T>0$, we will show that
 \begin{equation}\label{estimate:general:priori}
\displaystyle\sup_{t\in[0,T]}\left(\norm{\Omega}^{2}_{H^{4}}+ \norm{\theta}^{2}_{H^{2}} \right) \leq K(T)< \infty,
\end{equation}
where $K(T)=K\left(\norm{\xi_{k}}_{H^{3}}, \lambda, \norm{\Theta_{0}}_{H^{2}},\norm{\Omega_{0}}_{H^{4}}, T  \right)$, and $\lambda$ is the ellipticity constant specified in  \eqref{elliptic:Lie}. \\ \\
\textbf{$L^{2}$-estimate:} We begin by providing the $L^{2}$ estimate. We multiply the second equation in \eqref{new:expectation:vor:bouss} with $\Theta$ and integrate over $\TT^2$ to obtain
\[
\frac{1}{2}\frac{\diff}{\diff t}\norm{\Theta}^{2}_{L^2}+ \int_{\TT^2} \mathcal{L}_{\U}\Theta \Theta \ {\sf d} V =\displaystyle\sum_{k=1}^{N} \frac{1}{2}\int_{\TT^2} \mathcal{L}^{2}_{\xi_{k}}\Theta \Theta \ {\sf d} V. \]
Using the incompressibility condition, the second term on the left-hand side vanishes. Moreover, by expanding out the double Lie derivative operator, we have
\[ \displaystyle\sum_{k=1}^{N}\frac{1}{2}\int_{\TT^2} \mathcal{L}^{2}_{\xi_{k}}\Theta \Theta \ {\sf d} V = \frac{1}{2}\displaystyle\sum_{k=1}^{N} \displaystyle\sum_{i,j=1}^{2} \int_{\TT^2} \left( a_{k}^{ij}(x)\p_{i}\p_{j}\Theta+ b_{k}^{i}(x)\p_{i}\Theta \right) \Theta \  {\sf d} V:= \frac{1}{2}\left(I_{1}+I_{2}\right), \]
with coefficients $a_{k}^{ij}=\xi^{i}_{k}\xi^{j}_{k}$ and $b_{k}^{i}=\left(\xi_{k}\cdot\nabla\right)\xi^{i}_{k}$,
where
\begin{align*}
I_{1} &=-\displaystyle\sum_{k=1}^{N}\displaystyle\sum_{i,j=1}^{2}\int_{\TT^2} a_{k}^{ij}(x)(\p_{i}\Theta)(\p_{j}\Theta) \ {\sf d} V  - \displaystyle\sum_{k=1}^{N}\displaystyle\sum_{i,j=1}^{2}\int_{\TT^2} (\p_{i} a_{k}^{ij}(x))(\p_{j}\Theta) \Theta \ {\sf d} V, \\
I_2 &= \displaystyle\sum_{k=1}^{N}\displaystyle\sum_{i,j=1}^{2}\int_{\TT^2} b_{k}^{i}(x) (\p_{i}\Theta) \Theta {\sf d} V.
\end{align*}
This gives us
\[ \frac{1}{2}\frac{\diff}{\diff t}\norm{\Theta}^{2}_{L^2} +\displaystyle\sum_{k=1}^{N}\displaystyle\sum_{i,j=1}^{2} \frac{1}{2}\int_{\TT^2} a_{k}^{ij}(x)(\p_{i}\Theta) (\p_{j}\Theta) \ {\sf d} V =  -\displaystyle\sum_{k=1}^{N}\displaystyle\sum_{i,j=1}^{2} \frac{1}{2} \int_{\TT^2} (\p_{i} a_{k}^{ij}(x)) (\p_{j}\Theta) \Theta \ {\sf d} V + I_{2}, \]
and applying the uniform ellipticity condition \eqref{elliptic:Lie} and H\"{o}lder's inequality, we arrive at
\begin{equation}
\frac{1}{2}\frac{\diff }{\diff t}\norm{\Theta}^{2}_{L^2} + \frac{\lambda}{2}\int_{\TT^2} |\nabla\Theta|^{2} \ {\sf d} V \leq  \norm{\nabla a_{k}^{i,j}}_{L^\infty} \norm{\nabla\Theta}_{L^2} \norm{\Theta}_{L^2}+\norm{b_{k}^{i}}_{L^{\infty}}\norm{\nabla\Theta}_{L^2} \norm{\Theta}_{L^2},  
\end{equation}
for some $\lambda > 0$. Using Peter-Paul's inequality \eqref{PP:ine}, we find
\begin{equation*}
\frac{1}{2}\frac{\diff}{\diff t}\norm{\Theta}^{2}_{L^2} + \frac{\lambda}{2}\int_{\TT^2} |\nabla\Theta|^{2} \ {\sf d} V \leq   \left(\norm{\nabla a_{k}^{ij}}_{L^{\infty}}+\norm{b_{k}^i}_{L^{\infty}}\right)\left( \frac{\norm{\Theta}^{2}_{L^2}}{\epsilon}+\epsilon\norm{\nabla\Theta}^{2}_{L^2}  \right).
\end{equation*}
and choosing $\epsilon=\lambda/4\left(\norm{\nabla a_{k}^{ij}}_{L^{\infty}}+\norm{b_{k}^i}_{L^{\infty}}\right)$, we obtain
\begin{equation*}
\frac{1}{2}\frac{\diff}{\diff t}\norm{\Theta}^{2}_{L^2} + \frac{\lambda}{4}\int_{\TT^2} |\nabla\Theta|^{2} \ {\sf d} V \leq C_{0} \norm{\Theta}^{2}_{L^{2}},
\end{equation*} 
where $C_{0}=C\left(\lambda,\norm{\nabla a_{k}^{ij}}_{L^{\infty}},\norm{b_{k}^{i}}_{L^{\infty}}\right)$. By Gr\"{o}nwall's inequality, we conclude
\begin{equation}\label{L2:estimate:theta} \displaystyle\sup_{t\in[0,T]}\norm{\Theta}^{2}_{L^{2}} \leq 2C_{0}  \norm{\Theta_{0}}^{2}_{L^2}e^{T}, \ \text{and} \ \int_{0}^{T}\norm{\nabla\Theta}^{2}_{L^2} \ \diff \tau \leq 2C_{0}  \norm{\Theta_{0}}^{2}_{L^2}e^{T}.
\end{equation}

To compute the $L^{2}$ evolution of $\Omega$, we multiply the first equation in \eqref{new:expectation:vor:bouss} by $\Omega$ to get
\begin{align*}
\frac{1}{2}\frac{\diff}{\diff t}\norm{\Omega}^{2}_{L^2}+ \int_{\TT^2} (\mathcal{L}_{\U}\Omega) \Omega \ {\sf d} V = \int_{\TT^2} g(\partial_{x}\Theta)\Omega + \displaystyle\sum_{k=1}^{N}\frac{1}{2}\int_{\TT^2} (\mathcal{L}^{2}_{\xi_{k}} \Omega) \Omega \ {\sf d} V.
\end{align*}
By the incompresibility condition, the second term vanishes, and taking into account the ellipticity condition satisfied by the double Lie derivative operator \eqref{elliptic:Lie}, we integrate by parts and apply H\"older's inequality as before, obtaining
\[
\frac{1}{2}\frac{\diff}{\diff t}\norm{\Omega}^{2}_{L^2}+ \frac{\lambda}{2} \int_{\TT^2} |\nabla\Omega|^{2} \ {\sf d} V \leq \left(\|\nabla a_{k}^{i,j}\|_{L^\infty}+\|b_{k}^{i}\|_{L^\infty}\right)\norm{\Omega}_{L^{2}}\norm{\nabla\Omega}_{L^2}
+ g\norm{\Theta}_{L^2}\norm{\nabla\Omega}_{L^2}.
\]
Once again using Peter-Paul's inequality \eqref{PP:ine} and the fact that $\norm{\nabla \V}_{L^{2}}\leq  \norm{\Omega}_{L^{2}},$ by the Biot-Savart inequaliy \eqref{BiotSavart:ine} we obtain
\[ \frac{1}{2}\frac{\diff}{\diff t}||\Omega|^{2}_{L^2}+ \frac{\lambda}{4} \int_{\TT^2} |\nabla\Omega|^{2} \ {\sf d} V \leq C_{1}\left( \norm{\Omega}^{2}_{L^{2}}+\norm{\Theta}^{2}_{L^{2}}\right), \]
where $C_{1}=C\left(\lambda,\|\nabla a_{k}^{i,j}\|_{L^\infty},\norm{b_{k}^{i}}_{L^\infty}\right)$. Finally, by Gr\"onwall's inequality and bound \eqref{L2:estimate:theta}, we get
\begin{equation}\label{L2:estimate:omega}
\displaystyle\sup_{t\in[0,T]}\norm{\Omega}^{2}_{L^{2}} \leq C_{2} \left(\norm{\Omega_{0}}^{2}+\norm{\Theta_{0}}^{2}_{L^2}e^{T}\right)e^{T}, \ \text{and} \ \int_{0}^{T}\norm{\nabla \Omega}^{2}_{L^2} \ \diff \tau \leq C_{2} \left(\norm{\Omega_{0}}^{2}+\norm{\Theta_{0}}^{2}_{L^2}e^{T}\right)e^{T}.
 \end{equation}
with $C_{2}=C(C_{1},C_{0})$.
\\ \\
\textbf{$\dot{H}^{1}$-estimate:} Next, let us compute the $\dot{H}^1$-norm of $\Theta$. Integrating by parts, we have
\[ \frac{1}{2}\frac{\diff}{\diff t}\norm{\nabla\Theta}^{2}_{L^2}= \underbrace{-\int_{\TT^2} \nabla(\overbar{\U}\cdot\nabla\Theta)\cdot\nabla\Theta \ {\sf d} V}_{=0}
-\int_{\TT^2} \nabla(\V\cdot\nabla\Theta)\cdot\nabla\Theta \ {\sf d} V +\frac{1}{2} \displaystyle\sum_{k=1}^{N}\int_{\TT^2} \nabla \mathcal{L}^{2}_{\xi_{k}}\Theta\cdot\nabla\Theta \ {\sf d} V.
\]
For the second term on the right-hand side above, we obtain 
\begin{align*}
-\int_{\TT^2} \nabla(\V\cdot\nabla\Theta)\cdot\nabla\Theta {\sf d} V &= -\int_{\TT^2} (\nabla \V:\nabla\Theta)\cdot \nabla \Theta \ {\sf d} V - \frac{1}{2} \int_{\TT^2} \V\cdot \nabla(|\nabla\Theta|^{2}) {\sf d} V \\
&=-\int_{\TT^2} (\nabla \V:\nabla\Theta)\cdot \nabla \Theta \ {\sf d} V +\underbrace{\frac{1}{2} \int_{\TT^2} (\nabla\cdot \V) |\nabla\Theta|^{2} \ {\sf d} V}_{=0}. 
\end{align*}
Using H\"older's inequality, we have that
\begin{equation}
\abs{\int_{\TT^2} (\nabla \V:\nabla\Theta)\cdot \nabla \Theta \ {\sf d} V} \leq  ||\nabla \V||_{L^3}||\nabla \Theta||^{2}_{L^3}.
\end{equation}
For the third term, we get
\[\frac{1}{2}\displaystyle\sum_{k=1}^{N} \int_{\TT^2} \nabla \mathcal{L}_{\xi_{k}}^{2}\Theta\cdot\nabla\Theta \ {\sf d} V = \frac{1}{2}\displaystyle\sum_{k=1}^{N} \int_{\TT^2} \mathcal{L}_{\xi_{k}}^{2} \nabla \Theta\cdot\nabla\Theta \ {\sf d} V - \frac{1}{2} \displaystyle\sum_{k=1}^{N} \int_{\TT^2} \left[\mathcal{L}_{\xi_{k}}^{2} , \nabla \right]\Theta\cdot\nabla\Theta \ {\sf d} V .
\]
The commutator term can be bounded using H\"older's inequality
\[
\abs{\displaystyle \frac{1}{2} \sum_{k=1}^{N} \int_{\TT^2} \left[\mathcal{L}_{\xi_{k}}^{2} , \nabla \right]\Theta\cdot\nabla\Theta \ {\sf d} V}\leq \norm{\left[\mathcal{L}_{\xi_{k}}^{2} , \nabla \right]\Theta}_{L^2}\norm{\nabla\Theta}_{L^2} \leq C_{3}\norm{\nabla\Theta}^{2}_{L^{2}},
\]
since by a general result from harmonic analysis, the commutator $[\mathcal L^2_{\xi_k},\nabla]$ is a first order operator (cf. \cite{Tay76}). Here, the constant $C_3$  has dependence $C_{3}=C\left(\norm{\nabla a^{i,j}_{k}}_{L^{\infty}},\norm{b^{i}_{k}}_{L^{\infty}}\right)$.
As before, taking into account the uniform ellipticity condition, we arrive at
\begin{equation}\label{Main:eq:theta}
\frac{1}{2}\frac{\diff}{\diff t}\norm{\nabla\Theta}^{2}_{L^2}+ \frac{\lambda}{2}\int_{\TT^2} |\Delta \Theta|^{2} \ {\sf d} V \leq C_{3} \left(\norm{\nabla \V}_{L^3}\norm{\nabla \Theta}^{2}_{L^3} + \norm{\nabla\Theta}_{L^2}\norm{\Delta \Theta}_{L^2}+\norm{\nabla\Theta}^{2}_{L^2}\right).
\end{equation}
Moreover, by Peter-Paul's inequality \eqref{PP:ine} we have
\begin{equation}\label{PP:eq:theta}
\norm{\nabla\Theta}_{L^2}\norm{\Delta\Theta}_{L^2} \leq \frac{1}{2\delta}\norm{\nabla\Theta}^{2}_{L^2} + \frac{\delta}{2}\norm{\Delta\Theta}^{2}_{L^2,}
\end{equation}
for any $\delta>0$ and the bound
\begin{align}
||\nabla \V||_{L^3}\norm{\nabla\Theta}^{2}_{L^{3}} &\leq ||\nabla \V||^{3}_{L^{3}}+||\nabla \Theta||^{3}_{L^{3}} \nonumber \\
&\leq  ||\nabla \V||^{2}_{L^2}||\Delta \V||_{L^2} + ||\nabla \Theta||^{2}_{L^2}||\Delta \Theta||_{L^2} \nonumber \\
&\leq \frac{1}{2\epsilon}||\nabla \V||^{2}_{L^2}||\nabla \V||^{2}_{L^2} + \frac{\epsilon}{2} \norm{\Delta \V}^{2}_{L^2} + \frac{1}{2\nu}||\nabla \Theta||^{2}_{L^2}||\nabla \Theta||^{2}_{L^2} + \frac{\nu}{2} ||\Delta \Theta||^{2}_{L^2},\label{GN:eq:theta}
\end{align}
where $\epsilon>0$, $\nu>0$ will be chosen later on and we have invoked Gagliardo--Nirenberg inequality \eqref{Sob:ine1} in the second line and Peter--Paul's \eqref{PP:ine} inequality in the third one. Inserting \eqref{PP:eq:theta} and \eqref{GN:eq:theta} into  \eqref{Main:eq:theta}, we obtain that
\begin{align*}
 \frac{1}{2}\frac{\diff}{\diff t}\norm{\nabla\Theta}^{2}_{L^2}+ \frac{\lambda}{2}\int_{\TT^2} |\Delta \Theta|^{2} \ {\sf d} V &\leq C_{3} \Biggr( \frac{1}{2\epsilon}\norm{\nabla \V}^{2}_{L^2}\norm{\nabla \V}^{2}_{L^2} + \frac{\epsilon}{2} \norm{\Delta \V}^{2}_{L^2} \\
 & \ + \frac{1}{2\nu}||\nabla \Theta||^{2}_{L^2}||\nabla \Theta||^{2}_{L^2} + \frac{\nu}{2} ||\Delta \Theta||^{2}_{L^2}+\frac{1}{2\delta}\norm{\nabla\Theta}_{L^2}^2 + \frac{\delta}{2}\norm{\Delta\Theta}_{L^2}^2 \Biggr).
\end{align*}
Taking $\delta=\nu=\epsilon= \lambda /4C_{3}$, we have
\begin{equation*}
 \frac{1}{2}\frac{\diff}{\diff t}\norm{\nabla\Theta}^{2}_{L^2}+ \frac{\lambda}{4}\int_{\TT^2} |\Delta \Theta|^{2} \ {\sf d} V \leq C_{4} \left( \norm{\nabla \V}^{2}_{L^2}\norm{\nabla \V}^{2}_{L^2} +  \norm{\Delta \V}^{2}_{L^2} + ||\nabla \Theta||^{2}_{L^2}||\nabla \Theta||^{2}_{L^2} +\norm{\nabla\Theta}_{L^2}^2 \right),
\end{equation*}
with $C_{4}=C\left(C_{3},\lambda\right).$ Integrating in time, we obtain
\begin{align*}
&\norm{\nabla\Theta}^{2}_{L^2}+ \frac{\lambda}{4} \int_{0}^{t}\norm{\Delta \Theta}^{2}_{L^{2}} \diff\tau \\
&\leq C_{4} \left( \norm{\nabla\Theta_{0}}^{2}_{L^2}+ \int_{0}^{t}  \left(\norm{\nabla \V}^{2}_{L^2}\norm{\nabla \V}^{2}_{L^2} +  \norm{\Delta \V}^{2}_{L^2} + ||\nabla \Theta||^{2}_{L^2}||\nabla \Theta||^{2}_{L^2}  +\norm{\nabla\Theta}_{L^2}^2\right) \diff \tau \right).    
\end{align*} 

Noticing that $\norm{\nabla \V}_{L^{2}}\leq \norm{\Omega}_{L^{2}},$ $\norm{\Delta \V}_{L^{2}}\leq \norm{\nabla \Omega}_{L^{2}},$ and using the global bounds \eqref{L2:estimate:omega}, we see that by Gr\"onwall's inequality
\begin{equation}\label{global:grad:theta1}
\displaystyle\sup_{t\in[0,T]} \norm{\nabla\Theta}^{2}_{L^2}\leq C_{5} \norm{\nabla\Theta_{0}}^{2}_{L^2}\text{exp}\left(\int_{0}^{T} \left(\norm{\Omega_{0}}^{2}+\norm{\Theta_{0}}^{2}_{L^2}e^{t}\right)e^{t} \diff \tau\right) <\infty, \end{equation} 
and
\begin{equation}\label{global:grad:theta2}
\int_{0}^{T}\norm{\Delta \Theta}^2_{L^2} \diff \tau \leq C_{5}\norm{\nabla\Theta_{0}}^{2}_{L^2}\text{exp}\left(\int_{0}^{T} \left(\norm{\Omega_{0}}^{2}+\norm{\Theta_{0}}^{2}_{L^2}e^{t}\right)e^{t} \diff \tau \right)<\infty,
\end{equation}
where $C_{5}=C(C_{2},C_{4})$.
In a similar fashion, the evolution for the $\dot{H}^1$-norm of the vorticity $\Omega$ is given by
\begin{align*}
&\frac{1}{2}\frac{\diff}{\diff t}\norm{\nabla\Omega}^{2}_{L^2} - \frac{1}{2}\displaystyle\sum_{k=1}^{N} \int_{\TT^2} \mathcal{L}_{\xi_{k}}^{2} \nabla \Theta\cdot\nabla\Theta \ {\sf d} V \\
&=- \int_{\TT^2} \nabla(\overbar{\U}\cdot\nabla\Omega)\nabla\Omega \ {\sf d} V - \int_{\TT^2} \nabla(V\cdot\nabla\Omega)\nabla\Omega \ {\sf d} V
+\int_{\TT^2} \nabla \Theta_{x}\nabla\Omega \ {\sf d} V - \frac{1}{2} \displaystyle\sum_{k=1}^{N} \int_{\TT^2} \left[\mathcal{L}_{\xi_{k}}^{2} , \nabla \right]\Theta\cdot\nabla\Theta \ {\sf d} V\\
&=: K_1+K_2+K_3+K_4.
\end{align*}
Note that $K_1 = 0$ and the bounds for $K_2$ and $K_4$ can be obtained obtained in the same way as before
\begin{align*}
    |K_2| \leq \norm{\nabla \V}_{L^3}\norm{\nabla\Omega}^{2}_{L ^3}, \quad |K_4 | \leq  C_{6}\norm{\nabla\Omega}^{2}_{L^{2}},
\end{align*}
where $C_{6}=C\left(\norm{\nabla a^{i,j}_{k}}_{L^{\infty}},\norm{b^{i}_{k}}_{L^{\infty}}\right)$.
For $K_3$, integrating by parts and using H\"older's inequality, we get
\[ \abs{K_3}=\abs{\int_{\TT^2} \nabla \Theta \nabla\p_{x}
\Omega \ {\sf d} V} \leq ||\nabla \Theta||_{L^2}\norm{\Delta \Omega}_{L^2}. \]
The double Lie derivative term on the LHS can be manipulated in the same way as before using the ellipticity condition and we obtain
\begin{eqnarray*}
\frac{1}{2}\frac{\diff}{\diff t}\norm{\nabla\Omega}^{2}_{L^2}+\frac{\lambda}{2}\int_{\TT^2} |\Delta \Omega|^{2} \ {\sf d} V \leq C_6 \left( \norm{\nabla \V}_{L^3}\norm{\nabla \Omega}^{2}_{L^3}+\norm{\nabla\Omega}^{2}_{L^2}+ \norm{\nabla\Theta}_{L^2}\norm{\Delta\Omega}^{2}_{L^2} \right).
\end{eqnarray*}
Using Young's inequality, Gagliardo-Nirenberg \eqref{Sob:ine1}, and the fact that $\norm{\Delta \V}_{L^{2}}\leq \norm{\nabla \Omega}_{L^{2}}$ we get
\[
\frac{1}{2}\frac{\diff}{\diff t}\norm{\nabla\Omega}^{2}_{L^2}+\frac{\lambda}{4}\norm{\Delta \Omega}^{2}_{L^2}\leq C_{7} \left( \norm{\nabla\Theta}^{2}_{L^{2}}+\norm{\nabla\Omega}^{2}_{L^{2}} \right)\norm{\nabla\Theta}^{2}_{L^{2}},
\]
where $C_{7}=C\left(C_{6},\lambda\right)$. Integrating above in time, using Gr\"onwall's inequality, and noticing that we have the global bounds \eqref{L2:estimate:omega} and \eqref{global:grad:theta1} 
we find that
\[\displaystyle\sup_{t\in[0,T]}\norm{\nabla\Omega}^{2}_{L^2}+ \frac{\lambda}{4}\int_{0}^{T}\norm{\Delta \Omega}^{2}_{L^2} \diff t \leq C_{8}\norm{\nabla\Omega_{0}}_{L^2
}^{2}\exp{\left(\left(\norm{\Omega_{0}}^{2}+\norm{\Theta_{0}}^{2}_{L^2}e^{T}\right)e^{T}\right)}<\infty,
\]
where $C_{8}=C(C_{7},C_{2})$. One can check that the higher order estimates $\dot{H}^{4}$ and $\dot{H}^{3}$ for $\Omega$ and $\Theta$ respectively can be established in a similar way. To avoid repetition, we will not present the computations here. Thus, we have shown the a priori estimate \eqref{estimate:general:priori}. 

\textbf{Step 2: Mean growth control.} Next, let us control the growth of the mean part of $\U$ in terms of $\V,\Theta$ and the initial mean value $\overbar{\U}_{0}$. This estimate is essential in order to provide the uniqueness of solutions, as we will see later. Integrating \eqref{mean:equation:evolution} in time, we get (set $g=1$ without loss of generality)
\begin{equation}
  \overbar{\U}(t)=\overbar{\U}_{0} + \int_{0}^{t}\int_{\TT^{2}} \Theta \hat{y} \ {\sf d} V \ \diff \tau   + \frac{1}{2} \int_{0}^{t}  \displaystyle\sum_{k=1}^{N} \int_{\TT^{2}}  \mathcal{L}^{2}_{\xi_k}\V \ {\sf d} V \diff \tau  + \frac{1}{2} \int_{0}^{t}\displaystyle\sum_{k=1}^{N} \int_{\TT^{2}} \mathcal{L}^{2}_{\xi_k}\overbar{\U} \ {\sf d} V \diff \tau.
\end{equation} 
Using the fact that $\overbar{\U}$ does not depend on the spatial variable, we have
\[\frac{1}{2} \int_{0}^{t} \abs{ \displaystyle\sum_{k=1}^{N} \int_{\TT^{2}} \mathcal{L}^{2}_{\xi_k}\overbar{\U} \ {\sf d} V}  \rm{d}\tau \leq C_{9} \int_{0}^{t} \abs{\overbar{\U}} \rm{d}\tau, \quad \frac{1}{2} \int_{0}^{t} \abs{ \displaystyle\sum_{k=1}^{N} \int_{\TT^{2}} \mathcal{L}^{2}_{\xi_k}\V \ {\sf d} V}  \rm{d}\tau \leq C_{10} \int_{0}^{t} \norm{\V}_{L^{2}} \rm{d}\tau,
\]
where $C_{9}=C(\norm{\xi_{k}}^{2}_{H^{3}})<\infty$ and $C_{10}=C(\norm{\xi_{k}}^{2}_{H^{3}})<\infty$. Hence, we deduce
\[\abs{\overbar{\U}(t)}\leq C_{11} \abs{\overbar{\U}_{0}}+ \int_{0}^{t} \left(\norm{\V}_{L^{2}}+\norm{\Theta}_{L^{2}}+\abs{\overbar{\U}}\right) \rm{d}\tau, \]
with $C_{11}=C(C_{9},C_{10})$. By invoking Gr\"onwall's inequality, we conclude that
\begin{equation}\label{final:mean:estimate}
\displaystyle\sup_{t\in[0,T]} \abs{\overbar{\U}(t)} \leq C_{11} \left( \abs{\overbar{\U}_{0}} \left(\displaystyle\sup_{t\in[0,T]}\norm{\V}_{L^{2}}+\displaystyle\sup_{t\in[0,T]}\norm{\Theta}_{L^{2}}\right)e^T \right).
\end{equation}
\begin{remark}
Notice that from the equation for $\overbar{\U}$, we also have that the same equation holds for the differences $\widetilde{\overbar{\U}}=\overbar{\U}^{1}-\overbar{\U}^{2},$ i.e.
\begin{equation}\label{final:mean:estimate2}
\displaystyle\sup_{t\in[0,T]} \abs{\tilde{\overbar{\U}}(t)} \leq C_{11} \left( \abs{\tilde{\overbar{\U}}_{0}} \left(\displaystyle\sup_{t\in[0,T]}\norm{\tilde{\V}}_{L^{2}}+\displaystyle\sup_{t\in[0,T]}\norm{\tilde{\Theta}}_{L^{2}}\right)e^T \right).
\end{equation}
\end{remark}

\textbf{Step 3: Uniqueness of solutions.} To show uniqueness of solutions, we will prove that any two different solutions of \eqref{new:expectation:vor:bouss} with the same initial data must be equal.  As usual, we demonstrate it by deriving an estimate for the evolution of their difference and invoking Gr\"onwall's inequality. Let $\Omega^{1},\Omega^{2}$, $\Theta^{1},\Theta^{2}, \overbar{\U}^{1},\overbar{\U}^{2}$ be two solutions to \eqref{new:expectation:vor:bouss}-\eqref{mean:equation:evolution}, with the same initial data $(\Omega_{0},\Theta_{0},\overbar{\U}_{0})$. Defining the differences $\tilde{\Omega}:=\Omega^{1}-\Omega^{2},\tilde{\Theta}:=\Theta^{1}-\Theta^{2}, \tilde{\U}:=\U^{1}-\U^{2}$, we infer that
\begin{equation}
\left\{
\begin{array}{rl}
 \partial_{t}\tilde{\Omega} + \mathcal{L}_{\U^{1}}\tilde{\Omega}+ \mathcal{L}_{\tilde{\U}}\Omega^{2}&= g\p_{x}\tilde{\Theta} + \displaystyle \frac{1}{2} \sum_{k=1}^{N}  \mathcal{L}^{2}_{\xi_k}\tilde{\Omega}, \\ \partial_{t}\tilde{\Theta}+\mathcal{L}_{\U^{1}}\tilde{\Theta}+\mathcal{L}_{\tilde{\U}}\Theta^{2}&=\displaystyle \frac{1}{2} \sum_{k=1}^{N} \mathcal{L}^{2}_{\xi_k}\tilde{\Theta}.
 \end{array}
\right.
\end{equation}
In order to control the nonlinear terms, we apply the identity
\[ \int_{\TT^2} \left(\U^{1}\cdot \nabla\tilde{f}+\tilde{\U}\cdot\nabla f^{2}\right) \tilde{f} \ {\sf d} V= -\int_{\TT^2} \tilde{\U}\cdot\nabla\tilde{f}f^{2} \ {\sf d}V,\]
with $f = \Omega$ and $\Theta$, and using H\"older's, Peter-Paul's inequality \eqref{PP:ine}, and the Sobolev embedding \eqref{Sob:ine2}, we obtain
\begin{eqnarray}\label{uni:nonlinear1}
\abs{\int \tilde{\U}\cdot\nabla\tilde{f}f^{2} \ {\sf d} V} &\leq& ||\tilde{\U}||_{L^4}||\nabla\tilde{f}||_{L^2}||f^{2}||_{L^4} \leq \frac{1}{2\epsilon}||\tilde{\U}||^{2}_{L^4}||f^{2}||_{L^4}^{2}+  \frac{\epsilon}{2}||\nabla\tilde{f}||^{2}_{L^2} \nonumber \\
&\leq& \frac{1}{2\epsilon}||\tilde{\U}||^{2}_{H^1}||f^{2}||^{2}_{H^1}+\frac{\epsilon}{2}||\nabla\tilde{f}||^{2}_{L^2}. 
\end{eqnarray}

Following the a priori estimates we derived earlier, using \eqref{uni:nonlinear1} with $\epsilon=\lambda/4\left(\norm{\nabla a^{ij}_{k}}_{L^{\infty}}+\norm{b^{i}_{k}}_{L^{\infty}}\right)$ and integrating in time, we obtain
\[
||\tilde{\Omega}|^{2}_{L^2}+||\tilde{\Theta}|^{2}_{L^2} \leq C_{11} \left( ||\tilde{\Omega}_{0}|^{2}_{L^2}+||\tilde{\Theta}_{0}|^{2}_{L^2}+\int_{0}^{t} \left\{ ||\tilde{\Omega}||^{2}_{L^{2}}+||\tilde{\Theta}||^{2}_{L^{2}}+ ||\tilde{\U}||^{2}_{H^1}\left(||\Omega^{2}||^{2}_{H^1}+||\Theta^{2}||^{2}_{H^1}\right)\right \} \ \diff \tau \right),
\]
where $C_{11}=C\left(\lambda,\norm{\nabla a^{ij}_{k}}_{L^{\infty}},\norm{b^{i}_{k}}_{L^{\infty}}\right)$. Moreover, since $\tilde{\U}=\tilde{\overbar{\U}}+\tilde{\V}$, we have by \eqref{final:mean:estimate2} and the Biot--Savart inequality \eqref{BiotSavart:ine} that, 
\begin{eqnarray*}
 ||\tilde{\U}||^{2}_{H^1}=||\tilde{\overbar{\U}}+\tilde{\V}||^{2}_{L^2}+||\nabla \tilde{\V}||^{2}_{L^2} &\leq & |\tilde{\overbar{\U}}_{0}|^{2}+\displaystyle \sup_{t\in[0,T]}\left(\| \tilde{\V}\|_{L^2}+\|\tilde{\Theta}\|_{L^2})^{2}\right) +\|\nabla \tilde{\V}\|^{2}_{L^2} \\
&\lesssim & |\tilde{\overbar{\U}}_{0}|^{2}+\displaystyle \sup_{t\in[0,T]}\left( \|\tilde{\Omega}\|^{2}_{L^2}+\|\tilde{\Theta}\|^{2}_{L^2}\right). 
\end{eqnarray*}
Plugging this above and denoting $X(t):= \displaystyle\sup_{t\in[0,T]}\left(\|\tilde{\Omega}\|^{2}_{L^2}+\|\tilde{\Theta}\|^{2}_{L^2}\right)$ yields
\[
X(t) \leq X(0)+|\tilde{\overbar{\U}}_{0}|^{2}\int_{0}^{t} \left(||\Omega^{2}||^{2}_{H^1}+||\Theta^{2}||^{2}_{H^1}\right) \ \diff \tau + \int_{0}^{t} X(t) \left(||\Omega^{2}||^{2}_{H^1}+||\Theta^{2}||^{2}_{H^1}\right) \ \diff \tau,  \]
and using Gr\"onwall's inequality we obtain
\begin{eqnarray*}
X(t) &\leq & C_{12} \Biggl( X(0)+|\tilde{\overbar{\U}}_{0}|^{2}\int_{0}^{t}\left(||\Omega^{2}||^{2}_{H^1}+||\Theta^{2}||^{2}_{H^1}\right) \ \diff \tau \Biggr) \text{exp}\Big\{ C\int_{0}^{t} \left(||\Omega^{2}||^{2}_{H^1}+||\Theta^{2}||^{2}_{H^1}\right) \ \diff \tau \Big\} \\
&\lesssim &  \left(X(0)+|\tilde{\overbar{\U}}_{0}|^{2}\right) \text{exp}\Big\{||\Omega^{2}_{0}||^{2}_{H^1}+||\Theta^{2}_{0}||^{2}_{H^1}\Big\},
\end{eqnarray*}
for $C_{12}=C\left(K(T)\right)$.
Therefore, assuming that $\U^{1}_{0}=\U^{2}_{0}$ and $\Theta_{0}^{1}=\Theta_{0}^{2}$, we have that $\Omega^{1}=\Omega^{2}$ and therefore by the Biot-Savart embedding \eqref{BiotSavart:ine} , we find that $\U^{1}=\U^{2}$.

\begin{remark}
To establish the existence of strong solutions it suffices to apply a standard Galerkin approximation. Then we repeat the same a priori estimates and check that they are also satisfied by the approximated equations and independent of the truncation step. Then, by using Aubin--Lions compactness argument, we can pass to the limit and show that there exists a subsequence that converges to the desired strong solution. We do not provide the details here as this is standard in the PDE literature.
\end{remark}
\begin{remark}\label{parabolic-remark}
Moreover, due to the parabolic character of equations \eqref{new:expectation:vor:bouss} (cf. \cite{temam2001navier}) we can bootstrap the regularity and show that for any $t_{0}>0$, we have that $(\Omega,\Theta)\in C^{\infty}\left([t_{0},\infty)\times \TT^{2}\right).$
\end{remark}
\begin{remark}\label{recover:velocity}
Notice that there is a one-to-one correspondence between the solutions in terms of $(\U,\Theta)$ and $(\Omega,\bar{\U},\Theta)$. Indeed, since the velocity field $\U$ is incompressible, we have that $\U=\nabla^{\perp}\psi$, where $\psi$ is the stream function. On the other hand, we have that $\Omega= \text{curl} \ \U$ and hence $\Omega=-\Delta \psi$. Moreover, an explicit (non-local) relation between the vorticity and the mean-free part of the velocity field is provided by the Biot-Savart law $\V=\mathcal K \Omega$. An extra difficulty is that the spatial mean of $\U$ is not conserved and therefore, we need to solve for the mean part $\overbar{\U} := \int_{\TT^2} \U {\sf d} V$ separately in order to be able to fully reconstruct the velocity field $\U = \V + \overbar{\U}$. Hence, as we have shown above, once we know that the vorticity and the potential temperature are smooth, plus a good control of the evolution of the mean of the velocity field, we can recover the full velocity field and infer that it has the same regularity, i.e., smooth. 
\end{remark}
\begin{remark}\label{pressure:recovery}
 To recover the modified pressure term  
$ \nabla \left(P - \mathbb{E} \left[ \frac{|\u|^{2}}{2}\right] \right)$, we take the divergence in \eqref{new:expectation:bouss}, obtaining the  following Poisson equation 
\begin{equation}\label{eq:poisson}
-\Delta \widetilde{p}= \nabla\cdot \mathcal{L}_{\U}\U-g\p_{y}\Theta-\frac{1}{2} \sum_{k=1}^{N} \nabla\cdot \mathcal{L}^{2}_{\xi_k}\U,
 \end{equation}
 where we have used the incompressibility condition $\nabla\cdot \U=0$, denote $\widetilde{p}:=\left(P - \mathbb{E} \left[ \frac{|\u|^{2}}{2}\right] \right)$ and impose suitable periodic boundary conditions. This Poisson equation differs from the usual, since the double Lie derivative term on the right hand side does not vanish. Inverting the Laplacian and noticing that the RHS of equation \eqref{eq:poisson} is smooth, we can recover the modified pressure term which due to standard elliptic regularity estimates we infer that
\[ P - \mathbb{E} \left[ \frac{|\u|^{2}}{2}\right] \in C^{\infty}\left([t_{0},\infty)\times \TT^{2}\right).\] 
We avoid writing the explicit form of the modified pressure term $\widetilde{p}$, which is given by   the convolution with the periodic Newtonian potential, since we prefer not to include any kind of Sobolev or Lebesgue type estimates, cf. \cite{MajdaBertozzi}.
\end{remark}

\subsection{Well-posedness of the linear stochastic system}\label{Thm-1-proof}
We now show the proof of our main Theorem \eqref{main:th:1}. To that purpose, we use the fact that due to Theorem \ref{main:th:2} , we have that $\E{\u}\in C\left([0,\infty), H^5(\mathbb{T}^{2})\right)$. We divide the proof into two steps: first we show the existence of solutions by explicitly constructing them using the characteristics of the system. Next, we show the uniqueness of solutions by performing a standard energy estimate. 
\\ \\
\textbf{Step 1: Existence of solutions via characteristics.}
Consider the characteristics for the system \eqref{LA-Boussinesq-vorticity},
\begin{align}\label{char-eq}
    \begin{split}
    \diff \X_t &= \mathbb E[\u](t,\X_t) \diff t + \sum_{k=1}^N \xxi_k(t,\X_t) \circ \diff W_t^k \\
    &= \left(\mathbb E[\u](t,\X_t) + \frac12 \sum_{k=1}^{N} \sum_{j=1}^2 \xi_k^j \partial_j\xxi_k(t,\X_t)\right) \diff t + \sum_{k=1}^N \xxi_k(t,\X_t) \diff W_t^k.
    \end{split}
\end{align}
Invoking Theorem \ref{main:th:2}, we have that $\E{\u}\in C\left([0,\infty), C^{3,\alpha'}(\mathbb{T}^{2})\right)$ due to the Sobolev embedding, with $0<\alpha'<1$. Therefore \eqref{char-eq} admits a unique global solution by Picard's Theorem and its flow $\phi_{s,t}$ is $C^{3,\alpha}$ regular for every $0<\alpha < \alpha'$ (cf. \cite{kunita1997stochastic}).  

%By parabolicity of the expectation equation \eqref{expectation-eq}, $\mathbb E[u]$ admits as many (spacial) derivatives as we wish and furthermore is bounded in space due to Sobolev embedding. 

We claim that an explicit solution to the LA-Boussinesq system \eqref{LA-Boussinesq-vorticity} can be expressed as
\begin{align}\label{omega-theta-explicit}
\left\{
\begin{array}{rl}
    u(t,\x) &= (\phi_t)_* u_0 (\x) - g \displaystyle \int_{0}^{t} (\phi_{s,t})_*(y \,{\sf{d}}\theta)(s,\x) \diff s - \displaystyle\int_0^t {(\phi_{s,t})_*} f(s, \x) \diff s ,\\
    \theta(t,\x) &= \theta_0 (\phi_t^{-1}(\x)) ,
\end{array}
\right.
\end{align}
where we recall that $f(t,x) = -{\sf{d}}\mathbb{E} [p - |\u|^2/2-g\theta y] + g\E{\theta}\mb{\hat{y}}$, and we have employed the shorthand notation $\phi_t$ to represent $\phi_{0,t}$ and $(\phi_t)_*u_0$ to denote the push-forward of the one-form $u_0 = \u_0 \cdot \sf d \mb \x$ with respect to the flow $\phi_t,$ which is given explicitly by
\begin{align}\label{push}
    (\phi_* u)(\x) = \u_i(\phi^{-1}(\x)) \frac{\partial \left( \phi^{-1} \right)^i}{\partial x^j}(\x) {\sf{d}} x^j(\x).
\end{align}
Since the flow $\phi_{s,t}$ is global, this provides an explicit construction of a global solution to \eqref{LA:SALT:Stra:Bou}. Furthermore, one can verify that $(u_t,\theta_t) \in H^2(\mathbb T^2, \mathbb R) \times H^3(\mathbb T^2, \mathbb R)$ owing to the $C^{3,\alpha}$ regularity of the flow $\phi_{s,t}$. To show this, first note that the backward SDE for $\A_t(\x) := \phi_t^{-1}(\x)$ reads
\begin{align}\label{backward-eq}
    \diff \A_t = \left(-\mathbb E[\u](t,\A_t) + \frac12 \sum_{k=1}^N \sum_{j=1}^2 \xi_k^j \partial_j\xxi_k(t,\A_t)\right) \diff t - \sum_{k=1}^N \xxi_k(t,\A_t) \widehat{\diff W_t}^k,
\end{align}
where $\widehat{\diff W_t}$ denotes backward integration of the Brownian motion. Consider the mollification $\theta_0^\epsilon = \rho^\epsilon * \theta_0$. By It\^o's first formula \eqref{ito-first-formula}, we obtain
\begin{align*}
    \theta_0^\epsilon(\A_t(\x)) = \theta_0^\epsilon(\x) &- \int^t_0\left(\mathbb E[\u](s,\A_s) \cdot \nabla (\theta_0^\epsilon(\A_s(\x))) - \frac12 \displaystyle\sum_{k=1}^N \xxi_k \cdot \nabla \left(\xxi_k \cdot \nabla\right)(\theta_0^\epsilon(\A_s(\x))) \right) \diff s \\
    &- \displaystyle\sum_{k=1}^N \int^t_0 \xxi_k (s,\A_s) \cdot \nabla (\theta_0^\epsilon(\A_s(\x))) \diff W_s^k,
\end{align*}
and taking into account the density of smooth functions in $H^3$, we have strong convergence $\theta_0^\epsilon \rightarrow \theta_0$ as $\epsilon \rightarrow 0$, which implies that $\theta(t,\X) = \theta_0(\A_t(\X))$ solves the $\theta$-equation in \eqref{LA-Boussinesq-vorticity}. Furthermore, since the inverse map $\A_t$ has $C^{3,\alpha}$ regularity and the initial data $\theta_0$ is in $H^3$, we see that its composition is also in $H^3$. For the $u$-equation, again by applying It\^o's first formula and using the fact that smooth functions are dense in $H^2$, we have
\begin{align}
\begin{split}
(\phi_t)_*u_0(x) &= u_0(x) - \int^t_0 \LL_{\E{u}}\left((\phi_s)_*u_0\right)(x) \diff s + \frac12 \sum_{k=1}^N \int^t_0 \LL_{\xi_k}^2\left((\phi_s)_*u_0\right)(x) \diff s \\
&\quad - \sum_{k=1}^N\int^t_0 \LL_{\xi_k}\left((\phi_s)_*u_0\right)(x) \diff W_s^k.
\end{split}\label{u-1}
\end{align}

By Lemma \ref{ito-int}, we obtain
\begin{align}
    &\int^t_0(\phi_{s,t})_*(y{\sf{d}}\theta)(s,\x) \diff s = (\phi_t)_*\int^t_0\phi_s^*(y{\sf{d}}\theta)(s,\x) \diff s \nonumber\\
    \begin{split}
    &= \int^t_0 y{\sf{d}}\theta(s,\x)\diff s - \int^t_0 \LL_{\E{u}}\left(\int^s_0(\phi_{r,s})_*(y{\sf{d}}\theta)(r,\x)\diff r\right) \diff s \\
    & + \sum_{k=1}^N \int^t_0 \LL_{\xi_k}^2 \left(\int^s_0(\phi_{r,s})_*(y{\sf{d}}\theta)(r,\x)\diff r\right) \diff s - \sum_{k=1}^N \int^t_0 \LL_{\xi_k}\left(\int^s_0(\phi_{r,s})_*(y{\sf{d}}\theta)(r,\x)\right) \diff W_r,
    \end{split}\label{u-2}
\end{align}
and similarly,
\begin{align}
    \begin{split}
    &\int_0^t {(\phi_{s,t})_*} f(s, \x) \diff s = \int_0^t f(s, \x) \diff s - \int^t_0 \LL_{\E{u}}\left(\int_0^s {(\phi_{r,s})_*} f(r, \x) \diff r\right)\diff s \\
    &+ \frac12 \sum_{k=1}^N \int^t_0 \LL_{\xi_k}^2\left(\int_0^s {(\phi_{r,s})_*} f(r, \x) \diff r\right)\diff s - \sum_{k=1}^N\int^t_0 \LL_{\xi_k}\left(\int_0^s {(\phi_{r,s})_*} f(r, \x) \diff r\right)\diff W_s^k.
    \end{split}\label{u-3}
\end{align}

Combining \eqref{u-1},\eqref{u-2} and \eqref{u-3}, we see that indeed the expression for $u$ in \eqref{omega-theta-explicit} satisfies the $u$-equation in \eqref{LA-Boussinesq-vorticity}. Since $\phi_t^{-1}$ has $C^{3,\alpha}$ regularity and $\u_0$ is in $H^2$, the pushforward $(\phi_t)_*u_0$ is in $H^2$ by \eqref{push}. Similarly, the other terms on the RHS of \eqref{omega-theta-explicit} can be shown to be $H^2$ so $u_t$ is indeed in $H^2$.

\noindent{\bf Step 2: Uniqueness.} By linearity of the system, we need only to verify that $(u_t,\theta_t) \equiv (0,0)$ for all $t>0$ provided $(u_0,\theta_0) \equiv (0,0)$. The $L^2$ estimate for $\theta$ can be computed as:
\begin{align*}
    \|\theta_t\|_{L^2}^2 \lesssim  \|\theta_0\|_{L^2}^2 &+ \sum_{k=1}^N \int^t_0 \|\div (\xxi_k \div(\xxi_k))\|_{L^\infty}
    \|\theta_s\|_{L^2}^2 \diff s + \sup_{0<s<t}|M_s|,
\end{align*}
where we used the divergence-free condition for $\mathbb E[\u]$ and $M_t := \sum_{k=1}^{\infty} \int^t_0 \left<\div(\xxi_k), \theta^2 \right>_{L^2}(s) \diff W_s^k$. By Gr\"onwall's inequality, one finds
\begin{align*}
    \|\theta_t\|_{L^2}^2 \lesssim \|\theta_0\|_{L^2}^2 + \sup_{0<s<t}|M_s|.
\end{align*}
Taking the expectation on both sides, and using the Burkholder-Davis-Gundy inequality
\begin{align*}
    \sup_{0<s<t} \mathbb E[M_s] \lesssim \sum_{k=1}^{N} \mathbb E \left[\int^t_0 \left<\div(\xxi_k), \theta^2 \right>_{L^2}(s) \diff s \right] \leq \sum_{k=1}^{N} \int^t_0 \|\div(\xxi_k)\|_{L^\infty}\mathbb E[\|\theta_s\|_{L^2}^2] \diff s,
\end{align*}
we have
\begin{align*}
    \mathbb E[\|\theta_t\|_{L^2}^2] \lesssim \|\theta_0\|_{L^2}^2
\end{align*}
by Gr\"onwall,
which implies $\theta_t \equiv 0$ if $\theta_0 \equiv 0$. From this, we have that ${\sf{d}} \theta_t \equiv 0$ if $\theta_0 \equiv 0$ and by uniqueness of the expectation equations we have $f(t,x) = -{\sf{d}}\mathbb{E} [p - |\u|^2/2-g\theta y] + g\E{\theta}\mb{\hat{y}} \equiv 0$.
We can then deduce in a similar fashion that $u_t \equiv 0$ if $(u_0,\theta_0) \equiv (0,0)$.
% \textcolor{blue}{Diego's upshot: The $L^2$ estimate for $(\omega,\theta)$ is given by
% \[ \rm{d}\left( \norm{\omega}^{2}_{L^2}+\norm{\theta}^{2}_{L^2} \right) = - 2\displaystyle\sum_{i=1}^{N} \left( \langle \mathcal{L}_{\xi_{k}}\omega,\omega \rangle_{L^2} + \langle \mathcal{L}_{\xi_{k}}\theta,\theta \rangle_{L^2} \right) \ \diff W^{k}_{s} \]
% Integrating and taking expectation, we get to the same conclusion  (I think that the intial datum should be also in expectation).
% }

\bibliography{biblio}
\bibliographystyle{alpha}

\appendix
\section{Calculation of the Stratonovich-to-It\^o correction term in Proposition \ref{covariance-eq}} \label{strat-ito-correction}
In the proof of Proposition \ref{covariance-eq}, we converted the Stratonovich integral (here, we are taking the number of noise fields to be $N=1$ for simplicity)
\begin{align} \label{strat-int}
    \int^t_0 \left(\left(\LL_{\xi} a\right) \otimes a' + a' \otimes \left(\LL_{\xi} a\right)\right) \circ \diff W_s,
\end{align}
into an It\^o integral. Here, we will show how this is done for readers unfamiliar with stochastic calculus.

Consider a general stochastic process
\begin{align} \label{ito-process}
\diff X_t = \mu_t \diff t + \sigma_t \circ \diff W_t.
\end{align}
Then the Stratonovich integral $\int X_t \circ \diff W_t$ can be made into an It\^o integral by adding a cross-variance correction term,
\begin{align*}
    \int^t_0 X_s \circ \diff W_s = \int^t_0 X_s \diff W_s + \frac12 \left[X_\cdot, W_\cdot\right]_t,
\end{align*}
where, given a partition $0 = t_0 < t_1 <\ldots <t_N = 1$ with mesh size $\Delta t := \sup_{i \in [0,N]} |t_{i+1} - t_i|)$, the cross-variance is defined as the stochastic limit
\begin{align*}
    \left[X_\cdot, W_\cdot\right]_t := \lim_{\substack{\Delta t \rightarrow 0 \\ N\rightarrow \infty}} \sum_{i=1}^{N-1} (X_{t_{i+1}} - X_{t_i})(W_{t_{i+1}} - W_{t_i}),
\end{align*}
where the convergence is in probability. For processes of the form \eqref{ito-process}, one can check that the cross-variation simply reads
\begin{align} \label{cross-variance}
    \left[X_\cdot, W_\cdot\right]_t = \int^t_0 \sigma_s \diff s.
\end{align}

Now, let us apply this result to our Stratonovich integral \eqref{strat-int}. First we have that
\begin{align*}
    &\int^t_0 \left(\left(\LL_{\xi} a\right) \otimes a' + a' \otimes \left(\LL_{\xi} a\right)\right) \circ \diff W_s \\
    &\quad = \int^t_0 \left(\left(\LL_{\xi} a\right) \otimes a' + a' \otimes \left(\LL_{\xi} a\right) \otimes\right) \diff W_s + \frac12 \left(\left[\left(\LL_{\xi} a\right) \otimes a', W_{\cdot}\right]_t + \left[a' \otimes \left(\LL_{\xi} a\right), W_{\cdot}\right]_t\right),
\end{align*}
where $\left[\left(\LL_{\xi} a\right) \otimes a', W_{\cdot}\right]_t + \left[a' \otimes \left(\LL_{\xi} a\right), W_{\cdot}\right]_t$ is the cross-variance between the integrand of the Stratonovich integral \eqref{strat-int} with $W_t$.

Now, taking the Lie derivative $\LL_\xi$ on both sides of equation \eqref{a-trans-eq}, we get
\begin{align} \label{lie-der-a}
    \diff \,(\mathcal L_\xi a) + \mathcal L_\xi \mathcal L_{\E{u}} a \diff t + \LL_\xi (\LL_\xi a) \circ \diff W_t  = 0.
\end{align}
Using \eqref{lie-der-a}, \eqref{a-fluctuation-eq} and the stochastic Leibniz rule $\diff \,(S \otimes T) = (\circ \diff S) \otimes T + S \otimes (\circ \diff T)$, we get
\begin{align*}
    \diff (\left(\LL_{\xi} a\right) \otimes a') = -\left((\mathcal L_\xi \mathcal L_{\E{u}} a)\otimes a' + \left(\LL_{\xi} a\right) \otimes (\LL_{\E{u}} a')\right)\diff t - \left(( \LL_\xi \LL_\xi a) \otimes a' + (\LL_\xi a) \otimes (\LL_\xi a)\right)\circ \diff W_t.
\end{align*}
and similarly,
\begin{align*}
    \diff (a' \otimes \left(\LL_{\xi} a\right)) = -\left((\LL_{\E{u}} a') \otimes (\LL_{\xi} a) + a' \otimes (\mathcal L_\xi \mathcal L_{\E{u}} a)\right)\diff t - \left((\LL_\xi a) \otimes (\LL_\xi a) + a' \otimes  (\LL_\xi \LL_\xi a)\right)\circ \diff W_t.
\end{align*}
Then by \eqref{cross-variance}, we can conclude that
\begin{align*}
    &\left[\left(\LL_{\xi} a\right) \otimes a', W_{\cdot}\right]_t = -\int^t_0 \left(\mathcal L_{\xi}(\mathcal L_{\xi} a) \otimes a' + \left(\mathcal L_{\xi} a\right) \otimes \left(\mathcal L_{\xi} a\right)\right) \diff s, \\
    &\left[a' \otimes \left(\LL_{\xi} a\right), W_{\cdot}\right]_t = -\int^t_0 \left(\left(\mathcal L_{\xi} a\right) \otimes \left(\mathcal L_{\xi} a\right) + a' \otimes \mathcal L_{\xi}(\mathcal L_{\xi} a) \right)\diff s.
\end{align*}
Therefore, the Stratonovich integral \eqref{strat-int} in It\^o form reads.
\begin{align*}
    -\frac12 \int^t_0 \left(\left(\mathcal L_{\xi}(\mathcal L_{\xi} a)\right) \otimes a' + 2\left(\mathcal L_{\xi} a\right) \otimes \left(\mathcal L_{\xi} a\right) + a' \otimes \mathcal L_{\xi}(\mathcal L_{\xi} a)\right) \diff s + \int^t_0 \left(\left(\LL_{\xi} a\right) \otimes a' + a' \otimes \left(\LL_{\xi} a\right) \otimes\right) \diff W_s.
\end{align*}

\section{$p$-th central moments for advected scalar fields.} \label{p-th-moment-scalar}
Here, we will show how to derive the $p$-th central moment equation for scalar fields \eqref{pth-moment} in Remark \ref{p-th-moment-remark}. Consider the stochastic advection equation \eqref{a-trans-eq} for a scalar field $a$ and let $a' := a-\E{a}$, which satisfies \eqref{a-fluctuation-eq}.
Using It\^o' formula, one can check that
\begin{align}\label{a-p-fluc-eq}
    \begin{split}
    &\diff \,(a')^p + \mathcal L_{\E{u}}(a')^p\diff t + \displaystyle\sum_k p(a')^{p-1} \mathcal L_{\xi_k} a \diff W_t^k \\
    &= \frac{p}{2} \displaystyle\sum_k \left((p-1)(a')^{p-2}(\mathcal L_{\xi_k} a)^2 + (a')^{p-1} \mathcal L_{\xi_k}\mathcal L_{\xi_k} a'\right)\diff t.
    \end{split}
\end{align}
Now, by the Leibniz property of the Lie derivative, we have
\begin{align*}
    &(a')^{p-1} \mathcal L_{\xi_k}\mathcal L_{\xi_k} a' = \frac1p \mathcal L_{\xi_k}\mathcal L_{\xi_k} (a')^p - (p-1)(a')^{p-2} (\mathcal L_{\xi_k} a')^2 \\
    &=\frac1p \mathcal L_{\xi_k}\mathcal L_{\xi_k} (a')^p - (p-1)(a')^{p-2} (\mathcal L_{\xi_k} a)^2 + 2(p-1)(a')^{p-2} (\mathcal L_{\xi_k} a)(\mathcal L_{\xi_k} \E{a}) \\
    &\qquad - (p-1)(a')^{p-2} (\mathcal L_{\xi_k} \E{a})^2,
\end{align*}
which we can substitute in the last term of \eqref{a-p-fluc-eq} to get
\begin{align}\label{a-p-fluc-eq-2}
    \begin{split}
    &\diff \,(a')^p + \mathcal L_{\E{u}}(a')^p\diff t + \displaystyle\sum_k p(a')^{p-1} \mathcal L_{\xi_k} a \diff W_t^k \\
    &= \frac{p}{2} \displaystyle\sum_k \left(\frac1p \mathcal L_{\xi_k}\mathcal L_{\xi_k} (a')^p + 2 (\mathcal L_{\xi_k} (a')^{p-1})(\mathcal L_{\xi_k} \E{\omega}) + (p-1)(a')^{p-2} (\mathcal L_{\xi_k} \E{a})^2\right)\diff t.
    \end{split}
\end{align}

Finally, taking expectations on both sides of \eqref{a-p-fluc-eq-2} give us an iterative PDE for the $p$-th central moment of the advected scalar field $a$:
\begin{align*}
    \begin{split}
    \partial_{t} A^{(p)} + \mathcal L_{\E{u}}A^{(p)} &= \displaystyle\sum_k \left(\frac12 \mathcal L_{\xi_k}\mathcal L_{\xi_k} A^{(p)} + p\left(\mathcal L_{\xi_k} A^{(p-1)}\right)\left(\mathcal L_{\xi_k} \E{a}\right) + \frac{p(p-1)}{2}A^{(p-2)} \left(\mathcal L_{\xi_k} \E{a}\right)^2\right),
    \end{split}
\end{align*}
where $A^{(p)} := \E{(a-\E{a})^p}$.
\end{document}